\documentclass[12pt]{article}
\usepackage{amsmath}
\usepackage{sw20jart}
\usepackage{epic}
\usepackage{curves}

\newtheorem{theorem}{Theorem}[section]

\newtheorem{corollary}[theorem]{Corollary}

\newtheorem{proposition}[theorem]{Proposition}
\newtheorem{remark}[theorem]{Remark}

\newenvironment{proof}[1][Proof]{\noindent\textbf{#1.} }{\ \rule{0.5em}{0.5em}}
\input{tcilatex}

\begin{document}

\author{Mark Korenblit \\
Department of Computer Science\\
Holon Institute of Technology, Israel\\
korenblit@hit.ac.il\bigskip \\
Vadim E. Levit\\
Department of Computer Science and Mathematics\\
Ariel University, Israel\\
levitv@ariel.ac.il }
\title{Fibonacci Graphs and their Expressions}
\date{}
\maketitle

\begin{abstract}
The paper investigates relationship between algebraic expressions and
graphs. We consider a digraph called a Fibonacci graph which gives a generic
example of non-series-parallel graphs. Our intention in this paper is to
simplify the expressions of Fibonacci graphs and eventually find their
shortest representations. With that end in view, we describe the number of
methods for generating Fibonacci graph expressions and carry out their
comparative analysis.\medskip

\end{abstract}

\section{Introduction\label{intro}}

A \textit{graph }$G=(V,E)$ consists of a \textit{vertex set\ }$V$ and an 
\textit{edge set\ }$E$, where each edge corresponds to a pair $(v,w)$ of
vertices. If the edges are ordered pairs of vertices (i.e., the pair $(v,w)$
is different from the pair $(w,v)$), then we call the graph \textit{directed}
or\textit{\ digraph}; otherwise, we call it \textit{undirected}. If $(v,w)$
is an edge in a digraph, we say that $(v,w)$ \textit{leaves} vertex $v$ and 
\textit{enters} vertex $w$. In a digraph, the \textit{out-degree} of a
vertex is the number of edges leaving it, and the \textit{in-degree} of a
vertex is the number of edges entering it. A vertex in a digraph is a 
\textit{source} if no edges enter it, and a \textit{sink} if no edges leave
it.

A \textit{path} from vertex $v_{0}$ to vertex $v_{k}$ in a graph $G=(V,E)$
is a sequence of its vertices $\left[ v_{0},v_{1},v_{2},\ldots ,v_{k-1},v_{k}%
\right] $ such that $(v_{i-1},v_{i})\in E$ for $1\leq i\leq k$. $G$ is an 
\textit{acyclic graph} if there is no closed path $\left[ v_{0},v_{1},v_{2},%
\ldots ,v_{k},v_{0}\right] $ in $G$. A two-terminal directed acyclic graph (%
\textit{st-dag}) has only one source $s$ and only one sink $t$. In an
st-dag, every vertex lies on some path from $s$ to $t$.

A graph $G^{\prime }=(V^{\prime },E^{\prime })$ is a \textit{subgraph} of $%
G=(V,E)$ if $V^{\prime }\subseteq V$ and $E^{\prime }\subseteq E$. A graph $%
G $ is \textit{homeomorphic} to a graph $G^{\prime }$ (a \textit{homeomorph}
of $G^{\prime }$) if $G$ can be obtained by subdividing edges of $G^{\prime
} $ with new vertices.

We consider a \textit{labeled graph} which has labels attached to its edges.
Each path between the source and the sink (a \textit{sequential path}) in an
st-dag can be presented by a product of all edge labels of the path. We
define the sum of edge label products corresponding to all possible
sequential paths of an st-dag $G$ as the \textit{canonical expression }of $G$%
. An algebraic expression is called an \textit{st-dag expression} (a \textit{%
factoring of an st-dag} in \cite{BKS}) if it is algebraically equivalent to
the canonical expression of an st-dag. An st-dag expression consists of
terms (edge labels), the operators $+$ (disjoint union) and $\cdot $
(concatenation, also denoted by juxtaposition when no ambiguity arises), and
parentheses.

We define the \textit{complexity of an algebraic expression} in two ways.
The complexity of an algebraic expression is (i) the total number of terms
in the expression including all their appearances (\textit{the first
complexity characteristic}) or (ii) the number of plus operators in the
expression (\textit{the second complexity characteristic}). We will denote
the first and the second complexity characteristic of an st-dag expression
by $T(n)$ and $P(n)$, respectively, where $n$ is the number of vertices in
the graph (the \textit{size of the graph}).

An equivalent expression with the minimum complexity is called an \textit{%
optimal representation of the algebraic expression}.

A \textit{series-parallel} \textit{graph} is defined recursively as follows:

(i) A single edge $(u,v)$ is a series-parallel graph with source $u$ and
sink $v$.

(ii) If $G_{1}$ and $G_{2}$ are series-parallel graphs, so is the graph
obtained by either of the following operations:

\quad (a) Parallel composition: identify the source of $G_{1}$ with the
source of $G_{2}$ and the sink of $G_{1}$ with the sink of $G_{2}$.

\quad (b) Series composition: identify the sink of $G_{1}$ with the source
of $G_{2}$.

As shown in \cite{BKS} and \cite{KoL}, a series-parallel graph expression
has a representation in which each term appears only once. We proved in \cite%
{KoL} that this representation is an optimal representation of the
series-parallel graph expression from the perspective of the first
complexity characteristic. For example, the st-dag expression of the
series-parallel graph presented in Figure \ref{fig1} is $%
abd+abe+acd+ace+fe+fd$. Since it is a series-parallel graph, the expression
can be reduced to $(a(b+c)+f)(d+e)$, where each term appears once.

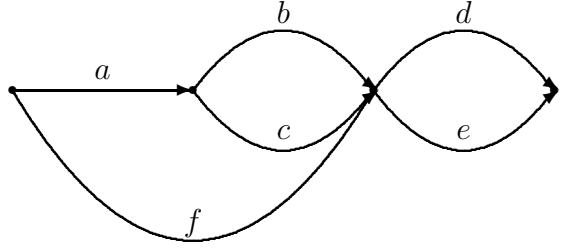
\begin{figure}[tbp]
\setlength{\unitlength}{0.8cm}
\par
\begin{picture}(5,4)(-3.5,0)\thicklines

\multiput(1,3)(3,0){4}{\circle*{0.15}}

\put(1,3){\vector(1,0){3}} \put(2.5,3.3){\makebox(0,0){$a$}}

\qbezier(4,3)(5.5,5)(7,3) \put(7.085,3){\vector(3,-2){0}}
\put(5.5,4.3){\makebox(0,0){$b$}}

\qbezier(4,3)(5.5,1)(7,3) \put(7.085,3){\vector(3,2){0}}
\put(5.5,2.3){\makebox(0,0){$c$}}

\qbezier(7,3)(8.5,5)(10,3) \put(10.085,3){\vector(3,-2){0}}
\put(8.5,4.3){\makebox(0,0){$d$}}

\qbezier(7,3)(8.5,1)(10,3) \put(10.085,3){\vector(3,2){0}}
\put(8.5,2.3){\makebox(0,0){$e$}}

\qbezier(1,3)(4,-2)(7,3) \put(7.085,3){\vector(4,3){0}}
\put(4,0.8){\makebox(0,0){$f$}}

\end{picture}
\caption{A series-parallel graph.}
\label{fig1}
\end{figure}

The notion of a \textit{Fibonacci graph }($FG$) was introduced in \cite{GoP}%
. A Fibonacci graph has vertices $\{1,2,3,\ldots ,n\}$ and edges 
\begin{equation*}
\left\{ \left( v,v+1\right) \mid v=1,2,\ldots ,n-1\right\} \cup \left\{
\left( v,v+2\right) \mid v=1,2,\ldots ,n-2\right\} .
\end{equation*}%
This graph is illustrated in Figure \ref{fig2}.

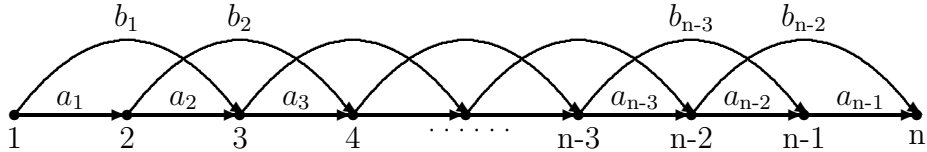
\begin{figure}[tbph]
\setlength{\unitlength}{1.0cm}
\par
\begin{picture}(5,2)(-0.9,-0.5)\thicklines

\multiput(0,0)(1.5,0){9}{\circle*{0.15}}

\put(0,-0.3){\makebox(0,0){1}}
\put(1.5,-0.3){\makebox(0,0){2}}
\put(3,-0.3){\makebox(0,0){3}}
\put(4.5,-0.3){\makebox(0,0){4}}
\put(7.5,-0.3){\makebox(0,0){n-3}}
\put(9,-0.3){\makebox(0,0){n-2}}
\put(10.5,-0.3){\makebox(0,0){n-1}}
\put(12,-0.3){\makebox(0,0){n}}

\multiput(0,0)(1.5,0){8}{\vector(1,0){1.5}}

\put(0.75,0.2){\makebox(0,0){$a_{1}$}}
\put(2.25,0.2){\makebox(0,0){$a_{2}$}}
\put(3.75,0.2){\makebox(0,0){$a_{3}$}}
\put(8.25,0.2){\makebox(0,0){$a_{\text{n-3}}$}}
\put(9.75,0.2){\makebox(0,0){$a_{\text{n-2}}$}}
\put(11.25,0.2){\makebox(0,0){$a_{\text{n-1}}$}}

\qbezier(0,0)(1.5,2)(3,0)
\qbezier(1.5,0)(3,2)(4.5,0)
\qbezier(3,0)(4.5,2)(6,0)
\qbezier(4.5,0)(6,2)(7.5,0)
\qbezier(6,0)(7.5,2)(9,0)
\qbezier(7.5,0)(9,2)(10.5,0)
\qbezier(9,0)(10.5,2)(12,0)

\multiput(3.085,0)(1.5,0){7}{\vector(3,-2){0}}

\put(1.5,1.3){\makebox(0,0){$b_{1}$}}
\put(3,1.3){\makebox(0,0){$b_{2}$}}
\put(9,1.3){\makebox(0,0){$b_{\text{n-3}}$}}
\put(10.5,1.3){\makebox(0,0){$b_{\text{n-2}}$}}

\multiput(5.55,-0.2)(0.2,0){6}{\circle*{0.02}}

\end{picture}
\caption{A Fibonacci graph.}
\label{fig2}
\end{figure}

As shown in \cite{Duf}, an st-dag is series-parallel if and only if it does
not contain a subgraph homeomorphic to the \textit{forbidden subgraph}
positioned between vertices $1$ and $4$ of the Fibonacci graph shown in
Figure \ref{fig2}. Thus, Fibonacci graphs are of interest as
\textquotedblright through\textquotedblright\ non-series-parallel st-dags.
Notice that Fibonacci graphs of size $2$ or $3$ are series-parallel.

Mutual relations between graphs and algebraic expressions are discussed in 
\cite{BKS}, \cite{GoM}, \cite{GMR}, \cite{KoL}, \cite{KoL1}, \cite{KoL2}, 
\cite{Mun1}, \cite{Mun2}, \cite{Nau}, \cite{SaW}, and other works.
Specifically, \cite{Mun1}, \cite{Mun2}, and \cite{SaW} consider the
correspondence between series-parallel graphs and read-once functions. A
Boolean function is defined as \textit{read-once} if it may be computed by
some formula in which no variable occurs more than once (\textit{read-once
formula}). On the other hand, a series-parallel graph expression can be
reduced to the representation in which each term appears only once. Hence,
such a representation of a series-parallel graph expression can be
considered to be a read-once formula (Boolean operations are replaced by
arithmetic ones).

An expression of a homeomorph of the forbidden subgraph belonging to any
non-series-parallel st-dag has no representation in which each term appears
once. For example, consider the subgraph positioned between vertices $1$ and 
$4$ of the Fibonacci graph shown in Figure \ref{fig2}. Possible optimal
representations of its expression are $a_{1}\left( a_{2}a_{3}+b_{2}\right)
+b_{1}a_{3}$ or $\left( a_{1}a_{2}+b_{1}\right) a_{3}+a_{1}b_{2}$. For this
reason, an expression of a non-series-parallel st-dag can not be represented
as a read-once formula. However, for arbitrary functions, which are not
read-once, generating the optimum factored form is NP-complete \cite{Wan}.
Some heuristic algorithms developed in order to obtain good factored forms
are described in \cite{GoM}, \cite{GMR} and other works. Therefore,
generating an optimal representation for a non-series-parallel st-dag
expression is a highly complex problem.

The problem of factoring boolean functions into shorter, more compact
formulae is one of the basic operations in algorithmic logic synthesis since
compactification saves money. In logic synthesis, one standard measure of
the complexity of a logic circuit is the number of terms. Computation time
also depends on the number of terms. However, computation time is determined
by the number of operations on terms as well. For this reason, the number of
plus operators is another important characteristic of a logic circuit.
Besides, the number of plus operators characterizes the number of
computation levels in a logic circuit (its \textquotedblright branching
out\textquotedblright\ degree).

Our intention in this paper is to simplify the expressions of Fibonacci
graphs (we denote them by $Ex(FG)$) and eventually find their optimal
representations. In \cite{KoL} we presented a heuristic algorithm with that
end in view and analyzed obtained expressions from the perspective of the
first complexity characteristic. Here we describe the number of methods for
generating Fibonacci graph expressions and carry out their comparative
analysis from the perspective of both the first and the second complexity
characteristics.

\section{Simple Methods\label{simple}}

This section considers three quite natural methods for generating
expressions of Fibonacci graphs.

\subsection{Sequential Paths Method\label{seq_paths}}

This method is based directly on the definition of an st-dag expression as
the canonical expression of the st-dag.

\begin{theorem}
\label{th_seq_exp}For an $n$-vertex $FG$:

1. The number of sequential paths $p(n)$ is defined recursively as follows: 
\begin{eqnarray}
p(1) &=&1  \notag \\
p(2) &=&1  \notag \\
p(n) &=&p(n-1)+p(n-2)\text{\quad }(n>2).  \label{fgf1}
\end{eqnarray}

2. The total number of terms $T(n)$ in the expression $Ex(FG)$ derived by
the sequential paths method is defined recursively as follows: 
\begin{eqnarray}
T(1) &=&0  \notag \\
T(2) &=&1  \notag \\
T(n) &=&T(n-1)+T(n-2)+p(n)\quad (n>2).  \label{fgf7}
\end{eqnarray}

3. The number of plus operators $P(n)$ in the expression $Ex(FG)$ derived by
the sequential paths method is defined recursively as follows: 
\begin{eqnarray}
P(1) &=&0  \notag \\
P(2) &=&0  \notag \\
P(n) &=&P(n-1)+P(n-2)+1\quad (n>2).  \label{fgf2}
\end{eqnarray}
\end{theorem}

\begin{proof}
1. Initial statements $p(1)=1$ and $p(2)=1$ follow clearly. All sequential
paths in a Fibonacci graph (see Figure \ref{fig2}) subdivide into two
groups. Paths of the first group start from the edge labeled $a_{1}$; paths
of the second group start from the edge labeled $b_{1}$. Paths of the first
group are all sequential paths of the $FG$ positioned between vertices $2$
and $n$ and are supplemented by an edge labeled $a_{1}$. This graph includes 
$n-1$ vertices, and, for this reason, the number of sequential paths in this
graph, and by extension, in the first group, is equal to $p(n-1)$. By
analogy, the number of sequential paths in the second group is equal to $%
p(n-2)$. Hence, the proof of the statement is complete.

2. Initial statements $T(1)=0$ and $T(2)=1$ follow clearly. Consider the
case of $n>2$. As was mentioned above, each sequential path of an $n$-vertex 
$FG$ is a sequential path of an $n-1$-vertex $FG$ or an $n-2$-vertex $FG$
which is supplemented by one edge. That is, each sequential path in an $n-1$%
-vertex $FG$ and an $n-2$-vertex $FG$ corresponds to an additional term in $%
T(n)$. Hence, 
\begin{eqnarray*}
T(n) &=&T(n-1)+p(n-1)+T(n-2)+p(n-2) \\
&=&T(n-1)+T(n-2)+p(n).
\end{eqnarray*}

3. Initial statements $P(1)=0$ and $P(2)=0$ follow clearly. Consider the
case of\textit{\ }$n>2$. Taking into consideration (\ref{fgf1}) and the
obvious equality $p(n)=P(n)+1$ for $n\geq 1$, formula (\ref{fgf2}) follows
immediately.\medskip
\end{proof}

\begin{remark}
\bigskip The number of sequential paths $p(n)$ in an $n$-vertex $FG$ is
equal to the Fibonacci number $F_{n}$ $\left( F_{1}=1,\text{ }F_{2}=1,\text{ 
}F_{n}=F_{n-1}+F_{n-2}\right) $.
\end{remark}

The following explicit formula for $F_{n}$ and, consequently, for $p(n)$ is
obtained by the method for linear recurrence relations solving \cite{Ros}
(henceforth, the method \cite{Ros}):%
\begin{equation}
p(n)=\frac{1}{\sqrt{5}}\left[ \left( \frac{1+\sqrt{5}}{2}\right) ^{n}-\left( 
\frac{1-\sqrt{5}}{2}\right) ^{n}\right] \text{\allowbreak }.  \label{fgf9}
\end{equation}

Using (\ref{fgf1}) and \ (\ref{fgf7}) the following recurrence for $T(n)$ is
derived:%
\begin{eqnarray}
T(1) &=&0  \notag \\
T(2) &=&1  \notag \\
T(3) &=&3  \notag \\
T(4) &=&7  \notag \\
T(n) &=&2T(n-1)+T(n-2)-2T(n-3)-T(n-4)\quad (n>4).  \label{fgf8}
\end{eqnarray}

\begin{corollary}
\label{cor_seq_exp}For an $n$-vertex $FG$:

1. The total number of terms $T(n)$ in the expression $Ex(FG)$ derived by
the sequential paths method is expressed explicitly as follows: 
\begin{equation*}
T(n)=\frac{1}{5}\left[ \left( \frac{1+\sqrt{5}}{2}n-\frac{3}{\sqrt{5}}%
\right) \left( \frac{1+\sqrt{5}}{2}\right) ^{n}+\left( \frac{1-\sqrt{5}}{2}n+%
\frac{3}{\sqrt{5}}\right) \left( \frac{1-\sqrt{5}}{2}\right) ^{n}\right] .
\end{equation*}

2. The number of plus operators $P(n)$ in the expression $Ex(FG)$ derived by
the sequential paths method is expressed explicitly as follows: 
\begin{equation*}
P(n)=\frac{1}{\sqrt{5}}\left[ \left( \frac{1+\sqrt{5}}{2}\right) ^{n}-\left( 
\frac{1-\sqrt{5}}{2}\right) ^{n}\text{\allowbreak }\right] -1.
\end{equation*}
\end{corollary}

\begin{proof}
\bigskip 1. The proof is based on relation (\ref{fg_fig8}) and on the method 
\cite{Ros}.

2. The proof follows immediately from (\ref{fgf9}) and the above mentioned
equality $p(n)=P(n)+1$.
\end{proof}

For $n=9$, the corresponding algebraic expression is 
\begin{eqnarray*}
&&a_{1}a_{2}a_{3}a_{4}a_{5}a_{6}a_{7}a_{8}+a_{1}a_{2}a_{3}a_{4}a_{5}a_{6}b_{7}+a_{1}a_{2}a_{3}a_{4}a_{5}b_{6}a_{8}+a_{1}a_{2}a_{3}a_{4}b_{5}a_{7}a_{8}+
\\
&&a_{1}a_{2}a_{3}a_{4}b_{5}b_{7}+a_{1}a_{2}a_{3}b_{4}a_{6}a_{7}a_{8}+a_{1}a_{2}a_{3}b_{4}a_{6}b_{7}+a_{1}a_{2}a_{3}b_{4}b_{6}a_{8}+
\\
&&a_{1}a_{2}b_{3}a_{5}a_{6}a_{7}a_{8}+a_{1}a_{2}b_{3}a_{5}a_{6}b_{7}+a_{1}a_{2}b_{3}a_{5}b_{6}a_{8}+a_{1}a_{2}b_{3}b_{5}a_{7}a_{8}+
\\
&&a_{1}a_{2}b_{3}b_{5}b_{7}+a_{1}b_{2}a_{4}a_{5}a_{6}a_{7}a_{8}+a_{1}b_{2}a_{4}a_{5}a_{6}b_{7}+a_{1}b_{2}a_{4}a_{5}b_{6}a_{8}+
\\
&&a_{1}b_{2}a_{4}b_{5}a_{7}a_{8}+a_{1}b_{2}a_{4}b_{5}b_{7}+a_{1}b_{2}b_{4}a_{6}a_{7}a_{8}+a_{1}b_{2}b_{4}a_{6}b_{7}+
\\
&&a_{1}b_{2}b_{4}b_{6}a_{8}+b_{1}a_{3}a_{4}a_{5}a_{6}a_{7}a_{8}+b_{1}a_{3}a_{4}a_{5}a_{6}b_{7}+b_{1}a_{3}a_{4}a_{5}b_{6}a_{8}+
\\
&&b_{1}a_{3}a_{4}b_{5}a_{7}a_{8}+b_{1}a_{3}a_{4}b_{5}b_{7}+b_{1}a_{3}b_{4}a_{6}a_{7}a_{8}+b_{1}a_{3}b_{4}a_{6}b_{7}+
\\
&&b_{1}a_{3}b_{4}b_{6}a_{8}+b_{1}b_{3}a_{5}a_{6}a_{7}a_{8}+b_{1}b_{3}a_{5}a_{6}b_{7}+b_{1}b_{3}a_{5}b_{6}a_{8}+
\\
&&b_{1}b_{3}b_{5}a_{7}a_{8}+b_{1}b_{3}b_{5}b_{7}.
\end{eqnarray*}%
It contains $34$ products (that correspond to $34$ sequential paths of the
graph), $201$ terms and $33$ plus operators.

\subsubsection{Time and Space Expenses of the Method}

The expression $Ex(FG)$ will be implemented in this and in other methods by
a linked list of the following characters: terms $a_{i}$ and $b_{i}$,
parentheses $"("$ and $")"$, and a sign $"+"$. Terms $a_{i}$ and $b_{i}$
conditionally considered as alone characters can be presented as character
sequences consisting of characters $"a"$ or $"b"$ and digits of number $i$.

We propose the following recursive algorithm which realizes the sequential
paths method:\medskip

$FG\_Sequential\_Paths(i,j,n,Arr,Expr)$

\begin{enumerate}
\item $\mathbf{if}$ $i$ $<$ $n$

\item \qquad $Arr[j]\longleftarrow a_{i}$

\item \label{seqp3}\qquad $FG\_Sequential\_Paths(i+1,j+1,n,Arr,Expr)$

\item \qquad $\mathbf{if}$ $i$ $<$ $n-1$

\item \qquad \qquad $Arr[j]\longleftarrow b_{i}$

\item \label{seqp6}\qquad \qquad $FG\_Sequential\_Paths(i+2,j+1,n,Arr,Expr)$

\item \label{seqp7}$\mathbf{else}$

\item \label{seqp8}$\qquad \mathbf{if}$ $\mathbf{not}$\textbf{\ }$\mathbf{%
empty}(Expr)$

\item \label{seqp9}$\qquad \qquad \mathbf{Insert\_to\_End}("+",Expr)$

\item \label{seqp10}$\qquad \mathbf{Copy\_to\_End}(Arr,j,Expr)$
\end{enumerate}

Given integers $i$ and $j$ and an auxiliary array $Arr$ of size $n-1$, this
procedure generates a linked list $Expr$ which implements the expression of
an $n$-vertex $FG$ derived by the sequential paths method. \ Array $Arr$ is
used to accumulate a product of terms corresponding to a current sequential
path. Integer $i$ is a number of a given vertex which edges labeled $a_{i}$
and $b_{i}$ leave, and $j$ is a subscript of an element in $Arr$. The
procedure is invoked with $i=j=1$ and an empty list $Expr$.

The procedure generates all admissible for our problem combinations of $%
a_{i} $ $(i=1,2,...,n-1)$ and $b_{i}$ $(i=1,2,...,n-2)$ \ After $i$ reaches $%
n$ (line \ref{seqp7}) a current combination has been composed and content of
the first $j$ elements of $Arr$ is copied at the end of $Expr$ (line \ref%
{seqp10}). If a derived product is not the first one in the expression,
i.e., $Expr$ is not an empty list, then a sign $"+"$ is inserted before the
product (lines \ref{seqp8}-\ref{seqp9}).

The running time of this algorithm consists of two components.

In the general case of the recursion, the running time, as follows from
lines \ref{seqp3} and \ref{seqp6}, is $t_{1}(n)=t_{1}(n-1)+t_{1}(n-2)+O(1)$,
i.e., the time complexity increases as Fibonacci numbers (see (\ref{fgf9})).

However, in the base case of the recursion, the time expenses are not
constant but are proportional to the size of a part of the expression copied
from $Arr$ to $Expr$. Therefore, this component of the algorithm's running
time is determined by the size of all the generated expression including all
terms and plus operators. As follows from Theorem \ref{th_seq_exp} and
Corollary \ref{cor_seq_exp}, the total number of terms is the most
significant part of the expression's size and, therefore, it states the
complexity of the time $t_{2}(n)$ for copying all parts of the expression
from $Arr$ to $Expr$. That is, by Corollary \ref{cor_seq_exp}, $%
t_{2}(n)=\Theta \left( n\left( \frac{1+\sqrt{5}}{2}\right) ^{n}\right) $.

Thus, the total running time of the algorithm is%
\begin{equation*}
t(n)=t_{1}(n)+t_{2}(n)=\Theta \left( \left( \frac{1+\sqrt{5}}{2}\right)
^{n}\right) +\Theta \left( n\left( \frac{1+\sqrt{5}}{2}\right) ^{n}\right)
=\Theta \left( n\left( \frac{1+\sqrt{5}}{2}\right) ^{n}\right) .
\end{equation*}

The algorithm uses only a $\Theta \left( n\right) $-size additional array
and so, the amount of memory it requires is determined by the size of the
derived expression and is also $\Theta \left( n\left( \frac{1+\sqrt{5}}{2}%
\right) ^{n}\right) $.

\subsection{Depth First Search (DFS) Method\label{dfs}}

An expression is derived by utilizing the well-known depth first search
algorithm \cite{CLR} and by using intermediate subexpressions which are
accumulated in st-dag's vertices. A subexpression which is accumulated in
vertex $i$ of the st-dag corresponds to its subgraph which is positioned
between vertices $i$ and $n$. The following recursive procedure is used:

\begin{enumerate}
\item \label{1dfs}The subexpression accumulated in vertex $n$ (see Figure %
\ref{fig2}) is equal to $1$.

\item The subexpression accumulated in vertex $n-1$ is equal to $a_{n-1}$.

\item The subexpression accumulated in vertex $i$ ($i<n-1$) is equal to $%
a_{i}E_{i+1}+b_{i}E_{i+2}$ where $E_{i+1}$ and $E_{i+2}$ are subexpressions
accumulated in vertices $i+1$ and $i+2$, respectively.

\item The subexpression accumulated in vertex $1$ is the resulting
expression.
\end{enumerate}

The special case of a subgraph consisting of a single vertex is considered
in line \ref{1dfs} of the recursive procedure. It is clear that such a
subgraph can be connected to other subgraphs only serially. For this reason,
it is accepted that its subexpression is $1$, so that, when it is multiplied
by another subexpression, the final result is not influenced.

\begin{theorem}
\label{th_dfs_exp}For an $n$-vertex $FG$:

1. The total number of terms $T(n)$ in the expression $Ex(FG)$ derived by
the DFS method is defined recursively as follows: 
\begin{eqnarray}
T(1) &=&0  \notag \\
T(2) &=&1  \notag \\
T(n) &=&T(n-1)+T(n-2)+2\quad (n>2).  \label{fgf3}
\end{eqnarray}

2. The number of plus operators $P(n)$ in the expression $Ex(FG)$ derived by
the DFS method is defined recursively as follows: 
\begin{eqnarray}
P(1) &=&0  \notag \\
P(2) &=&0  \notag \\
P(n) &=&P(n-1)+P(n-2)+1\quad (n>2).  \label{fgf4}
\end{eqnarray}
\end{theorem}

\begin{proof}
1. Initial statements $T(1)=0$ and $T(2)=1$ follow clearly. The resulting
expression $Ex(FG)$ is equal to $a_{1}E_{2}+b_{1}E_{3}$ where $E_{2}$ and $%
E_{3}$ are subexpressions accumulated in vertices $2$ and $3$, respectively
(see Figure \ref{fig2} and the DFS recursive procedure). $E_{2}$ is the
symbolic expression of the $FG$ which is positioned between vertices $2$ and 
$n$. This graph includes $n-1$ vertices and, for this reason, the total
number of terms in $E_{2}$ is equal to $T(n-1)$. By analogy, the total
number of terms in $E_{3}$ is equal to $T(n-2)$. Terms $a_{1}$ and $b_{1}$
are two additional terms in $Ex(FG)$. Hence, the proof of the statement is
complete.

2. This second proof is analogous to the first one. The expression $%
a_{1}E_{2}+b_{1}E_{3}$ includes all plus operations of $E_{2}$ and $E_{3}$
and one additional plus operation.\medskip
\end{proof}

\bigskip As follows from Theorems \ref{th_seq_exp} and \ref{th_dfs_exp}, a
method's evaluation depends on the kind of complexity that has been chosen.
If methods are compared by means of the second complexity characteristic,
then sequential paths and DFS methods are equivalent. However, from the
perspective of the first complexity characteristic, the DFS method is more
efficient.

\begin{corollary}
\label{col_dfs_exp}For an $n$-vertex $FG$:

1. The total number of terms $T(n)$ in the expression $Ex(FG)$ derived by
the DFS method is expressed explicitly as follows: 
\begin{equation*}
T(n)=\frac{1}{10}\left[ \left( 5+3\sqrt{5}\right) \left( \frac{1+\sqrt{5}}{2}%
\right) ^{n}+\left( 5-3\sqrt{5}\right) \left( \frac{1-\sqrt{5}}{2}\right)
^{n}\text{\allowbreak }\right] -2.
\end{equation*}

2. The number of plus operators $P(n)$ in the expression $Ex(FG)$ derived by
the DFS method is expressed explicitly as follows: 
\begin{equation*}
P(n)=\frac{1}{\sqrt{5}}\left[ \left( \frac{1+\sqrt{5}}{2}\right) ^{n}-\left( 
\frac{1-\sqrt{5}}{2}\right) ^{n}\text{\allowbreak }\right] -1.
\end{equation*}
\end{corollary}

\begin{proof}
\bigskip 1. The proof uses the recurrence obtained in Theorem \ref%
{th_dfs_exp} and is based on the method \cite{Ros}.

2. The proof follows immediately from Corollary \ref{cor_seq_exp} and the
equivalence of sequential paths and DFS methods from the perspective of the
second complexity characteristic.
\end{proof}

For $n=9$, the corresponding algebraic expression is 
\begin{eqnarray*}
&&a_{1}(a_{2}(a_{3}(a_{4}(a_{5}(a_{6}(a_{7}a_{8}+b_{7})+b_{6}a_{8})+b_{5}(a_{7}a_{8}+b_{7}))+
\\
&&b_{4}(a_{6}(a_{7}a_{8}+b_{7})+b_{6}a_{8}))+b_{3}(a_{5}(a_{6}(a_{7}a_{8}+b_{7})+b_{6}a_{8})+
\\
&&b_{5}(a_{7}a_{8}+b_{7})))+b_{2}(a_{4}(a_{5}(a_{6}(a_{7}a_{8}+b_{7})+b_{6}a_{8})+b_{5}(a_{7}a_{8}+b_{7}))+
\\
&&b_{4}(a_{6}(a_{7}a_{8}+b_{7})+b_{6}a_{8})))+b_{1}(a_{3}(a_{4}(a_{5}(a_{6}(a_{7}a_{8}+b_{7})+b_{6}a_{8})+
\\
&&b_{5}(a_{7}a_{8}+b_{7}))+b_{4}(a_{6}(a_{7}a_{8}+b_{7})+b_{6}a_{8}))+b_{3}(a_{5}(a_{6}(a_{7}a_{8}+b_{7})+b_{6}a_{8})+
\\
&&b_{5}(a_{7}a_{8}+b_{7}))).
\end{eqnarray*}%
It contains $87$ terms and $33$ plus operators.

Hence, this algorithm optimizes prefix parts of all subexpressions. In
principle, the DFS method can be applied by traversing the st-dag in the
opposite direction. In such a case, suffix parts of subexpressions are
optimized. Expression complexity characteristics will be the same.

\subsubsection{Time and Space Expenses of the Method\label{dfs_anal}}

We propose the following recursive algorithm which realizes the DFS\ method
in accordance with the above procedure:\medskip\ 

$FG\_DFS\_direct(i,n,Expr)$

\begin{enumerate}
\item $\mathbf{if}$ $i$ $<$ $n-1$

\item \label{dfsd2}\qquad $FG\_DFS\_direct(i+1,n,Expr)$

\item \qquad $\mathbf{if}$ $i$ $<$ $n-2$

\item $\qquad \qquad \mathbf{Insert\_to\_Head}\left( "(",Expr\right) $

\item $\qquad \qquad \mathbf{Insert\_to\_End}(")",Expr)$

\item \qquad $\mathbf{Insert\_to\_Head}\left( "a_{i}",Expr\right) $

\item \label{dfsd7}\qquad $FG\_DFS\_direct(i+2,n,Expr2)$

\item \qquad $\mathbf{if}$ $i$ $<$ $n-3$

\item \qquad \qquad $\mathbf{Insert\_to\_Head}\left( "(",Expr2\right) $

\item $\qquad \qquad \mathbf{Insert\_to\_End}(")",Expr2)$

\item \qquad $\mathbf{Insert\_to\_Head}\left( "+b_{i}",Expr2\right) $

\item \label{dfsd12}\qquad $\mathbf{Concatenate}(Expr,Expr2)$

\item $\mathbf{else}$

\item $\qquad \mathbf{if}$ $i$ $=$ $n-1$

\item \label{dfsd15}$\qquad \qquad Expr\leftarrow "a_{n-1}"$

\item $\qquad \mathbf{else}$

\item \label{dfsd17}$\qquad \qquad Expr\leftarrow \mathbf{NULL}$
\end{enumerate}

The algorithm generates a linked list $Expr$ which implements an expression
of a subgraph positioned between vertices $i$ and $n$. Parameter $i$ is
substituted by $1$ initially, for deriving the expression of an $n$-vertex $%
FG$.

Recursive calls in lines \ref{dfsd2} and \ref{dfsd7} of the algorithm
generate expressions for subgraphs with sources $i+1$ and $i+2$,
respectively. The first expression is presented as list $Expr$ and the
second one is presented as an additional list $Expr2$. After corresponding
insertions of terms $a_{i}$ and $b_{i}$, parentheses and a sign $"+"$, these
lists are concatenated in $O(1)$ time into the unified list $Expr$ (line \ref%
{dfsd12}) by assigning the address of the first element in $Expr2$ to the
pointer in the last element of $Expr$. In the base cases of the recursion, $%
Expr$ consists of the single term $a_{n-1}$ (line \ref{dfsd15}) or is an
empty list (line \ref{dfsd17}).

Thus, the running time of the algorithm is $t(n)=t(n-1)+t(n-2)+O(1)$, i.e.,
its complexity increases as Fibonacci numbers and $t(n)=\Theta \left( \left( 
\frac{1+\sqrt{5}}{2}\right) ^{n}\right) $.

The amount of memory that requires the algorithm is determined only by the
size of the derived expression and is also $\Theta \left( \left( \frac{1+%
\sqrt{5}}{2}\right) ^{n}\right) $.

The following algorithm realizes the DFS method applied in the opposite
direction:\medskip

$FG\_DFS\_opposite(i,n,Expr)$

\begin{enumerate}
\item $\mathbf{if}$ $i$ $<$ $n-1$

\item \qquad $FG\_DFS\_opposite(i,n-1,Expr)$

\item \qquad $\mathbf{if}$ $i$ $<$ $n-2$

\item \qquad \qquad $\mathbf{Insert\_to\_Head}("(",Expr)$

\item $\qquad \qquad \mathbf{Insert\_to\_End}\left( ")",Expr\right) $

\item $\qquad \mathbf{Insert\_to\_End}\left( "a_{n-1}+",Expr\right) $

\item \qquad $FG\_DFS\_opposite(i,n-2,Expr2)$

\item \qquad $\mathbf{if}$ $i$ $<$ $n-3$

\item \qquad \qquad $\mathbf{Insert\_to\_Head}("(",Expr2)$

\item $\qquad \qquad \mathbf{Insert\_to\_End}\left( ")",Expr2\right) $

\item $\qquad \mathbf{Insert\_to\_End}\left( "b_{n-2}",Expr2\right) $

\item \qquad $\mathbf{Concatenate}(Expr,Expr2)$

\item $\mathbf{else}$

\item $\qquad \mathbf{if}$ $i$ $=$ $n-1$

\item $\qquad \qquad Expr\leftarrow "a_{i}"$

\item $\qquad \mathbf{else}$

\item $\qquad \qquad Expr\leftarrow \mathbf{NULL}$
\end{enumerate}

It is clear that time and space expenses of this algorithm are the same as
of the previous one and are $\Theta \left( n\left( \frac{1+\sqrt{5}}{2}%
\right) ^{n}\right) $.

\subsection{Depth Last Search (DLS) Method\label{dls}}

While the DLS method is similar to the DFS method they differ in the fact
that an st-dag expression in the DLS method is derived by using special
subexpressions such as $a_{i}a_{i+1}+b_{i}$ which are related to
corresponding closed graph segments. The following recursive procedure is
used:

\begin{enumerate}
\item The subexpression accumulated in vertex $n$ is equal to $1$.

\item The subexpression accumulated in vertex $n-1$ is equal to $a_{n-1}$.

\item The subexpression accumulated in vertex $n-2$ is equal to $%
a_{n-2}a_{n-1}+b_{n-2}$.

\item The subexpression accumulated in vertex $i$ ($i<n-2$) is equal to $%
(a_{i}a_{i+1}+b_{i})E_{i+2}+a_{i}b_{i+1}E_{i+3}$ where $E_{i+2}$ and $%
E_{i+3} $ are subexpressions accumulated in vertices $i+2$ and $i+3$,
respectively.

\item The subexpression accumulated in vertex $1$ is the resulting
expression.
\end{enumerate}

\begin{theorem}
\label{th_dls_exp}For an $n$-vertex $FG$:

1. The total number of terms $T(n)$ in the expression $Ex(FG)$ derived by
the DLS method is defined recursively as follows: 
\begin{eqnarray*}
T(1) &=&0 \\
T(2) &=&1 \\
T(3) &=&3 \\
T(n) &=&T(n-2)+T(n-3)+5\quad (n>3).
\end{eqnarray*}

2. The number of plus operators $P(n)$ in the expression $Ex(FG)$ derived by
the DLS method is defined recursively as follows: 
\begin{eqnarray*}
P(1) &=&0 \\
P(2) &=&0 \\
P(3) &=&1 \\
P(n) &=&P(n-2)+P(n-3)+2\quad (n>3).
\end{eqnarray*}
\end{theorem}

\begin{proof}
1. Initial statements $T(1)=0$, $T(2)=1$, and $T(3)=3$ follow clearly. The
resulting expression $Ex(FG)$ is equal to $%
(a_{1}a_{2}+b_{1})E_{3}+a_{1}b_{2}E_{4}$, where $E_{3}$ and $E_{4}$ are
subexpressions accumulated in the vertices $3$ and $4$, respectively (see
Figure \ref{fig2} and the DLS recursive procedure). $E_{3}$ is the symbolic
expression of the $FG$ which is positioned between vertices $3$ and $n$.
This graph includes $n-2$ vertices, and for this reason, the total number of
terms in $E_{3}$ is equal to $T(n-2)$. By analogy, the total number of terms
in $E_{4}$ is equal to $T(n-3)$. Terms $a_{1}$, $a_{2}$, $b_{1}$, $a_{1}$,
and $b_{2}$ are five additional terms in $Ex(FG)$. Hence, the proof of the
statement is complete.

2. This second proof is analogous to the first one. The expression $%
(a_{1}a_{2}+b_{1})E_{3}+a_{1}b_{2}E_{4}$ includes all plus operations of $%
E_{3}$ and $E_{4}$ and two additional plus operations.\medskip
\end{proof}

As follows from Theorems \ref{th_seq_exp}, \ref{th_dfs_exp}, and \ref%
{th_dls_exp}, the DLS method is more efficient than sequential paths and DFS
methods from the perspective of both complexity characteristics.

\begin{corollary}
\label{cor_dls_exp}For an $n$-vertex $FG$:

1. The total number of terms $T(n)$ in the expression $Ex(FG)$ derived by
the DLS method is expressed explicitly as follows: 
\begin{eqnarray*}
T(n) &\approx &3.\,\allowbreak 4912\left( 1.\,\allowbreak 324\,7\right)
^{n}+\left( -0.245\,46-0.0\,\allowbreak 449\,7i\right) \left(
-0.662\,36+0.562\,28i\right) ^{n}+ \\
&&\left( -0.245\,46+0.\,0\allowbreak 449\,7i\right) \left(
-0.662\,36-0.562\,28i\right) ^{n}-5
\end{eqnarray*}

or%
\begin{eqnarray*}
T(n) &\approx &3.\,\allowbreak 491\,2\left( 1.\,\allowbreak 324\,7\right)
^{n}+ \\
&&\left( -1\right) ^{n+1\allowbreak }\left( 0.868\,84\right) ^{n}\left[
0.491\,09\cos \left( 0.703\,86n\right) +0.08894\,2\sin \left(
0.703\,86n\right) \right] -5.
\end{eqnarray*}

2. The number of plus operators $P(n)$ in the expression $Ex(FG)$ derived by
the DLS method is expressed explicitly as follows: 
\begin{eqnarray*}
P(n) &\approx &1.\,\allowbreak 267\,2\left( 1.\,\allowbreak 324\,7\right)
^{n}+\left( -0.133\,62-0.128\,28i\right) \left( -0.662\,36+0.562\,28i\right)
^{n}+ \\
&&\left( -0.133\,62+0.128\,28i\right) \left( -0.662\,36-0.562\,28i\right)
^{n}-2
\end{eqnarray*}

or%
\begin{eqnarray*}
P(n) &\approx &1.2672\left( 1.\,\allowbreak 324\,7\right) ^{n}+ \\
&&\left( -1\right) ^{n+1\allowbreak }\left( 0.868\,84\right) ^{n}\left[
0.26724\cos \left( 0.703\,86n\right) +0.25655\sin \left( 0.703\,86n\right) %
\right] -2.
\end{eqnarray*}
\end{corollary}

The proof of Corollary \ref{cor_dls_exp} uses the recurrences obtained in
Theorem \ref{th_dls_exp} and is based on the method \cite{Ros}.

For $n=9$, the corresponding algebraic expression is 
\begin{eqnarray*}
&&(a_{1}a_{2}+b_{1})((a_{3}a_{4}+b_{3})((a_{5}a_{6}+b_{5})(a_{7}a_{8}+b_{7})+a_{5}b_{6}a_{8})+
\\
&&a_{3}b_{4}((a_{6}a_{7}+b_{6})a_{8}+a_{6}b_{7}))+a_{1}b_{2}((a_{4}a_{5}+b_{4})((a_{6}a_{7}+b_{6})a_{8}+a_{6}b_{7})+
\\
&&a_{4}b_{5}(a_{7}a_{8}+b_{7})).
\end{eqnarray*}%
It contains $39$ terms and $14$ plus operators.

Like the DFS method, the DLS method can be employed by traversing the $FG$
in the opposite direction.

\subsubsection{Time and Space Expenses of the Method}

The algorithms which realize the method applied in both the direct and the
opposite directions are similar to the above algorithms for the DFS method
(section \ref{dfs_anal}). Complexities of running time and memory required
by these algorithms are defined as in the previous methods by the same
recurrences as the expression size and, by Corollary \ref{cor_dls_exp}, are
about $\Theta \left( \left( 1.\,\allowbreak 324\,7\right) ^{n}\right) $.

\section{Reduction Method\label{reduct}}

This method is based on the idea of \textit{reduction}. A \textit{series
reduction} at vertex $v$ is possible when $a=(u,v)$ is the unique edge
entering $v$, and $b=(v,w)$ is the unique edge leaving $v$: then $a$ and $b$
are replaced by $(u,w)$. A \textit{parallel reduction} at vertices $v$, $w$
replaces two or more edges $a_{1},\ldots ,a_{k}$ joining $v$ to $w$ by a
single edge $(v,w)$. A \textit{node reduction} at $v$ can occur when $v$ has
in-degree or out-degree $1$ (a node reduction is a generalization of a
series reduction). Suppose $v$ has in-degree $1$, and let $a=(u,v)$ be the
edge entering $v$. Let $b_{1}=(v,w_{1}),\ldots ,b_{k}=(v,w_{k})$ be the
edges leaving $v$. Replace $\{a,b_{1},\ldots ,b_{k}\}$ by $\{g_{1},\ldots
,g_{k}\}$, where $g_{i}=(u,w_{i})$. We call such a reduction a \textit{fork
reduction}. The case where $v$ has out-degree $1$ is symmetric: here $%
a=(v,w) $, $b_{i}=(u_{i},v)$, and $g_{i}=(u_{i},w)$. We call such a
reduction a \textit{joint reduction}.

The algorithm for generating an st-dag expression from a sequence of series,
parallel, and node reductions is proposed in \cite{BKS}. From a sequence of
series, parallel, and node reductions, reducing an arbitrary st-dag $G$ to a
single edge, an expression of this st-dag can be obtained as follows. We
denote, for the sake of brevity, the label of every edge $\mathbf{e}$ before
every reduction as $\mathbf{e}$. Then, the new label for the edge resulting
from a series reduction of $a$ and $b$ is $ab$. For a parallel reduction,
the new edge is labeled $a_{1}+\ldots +a_{k}$. The new label for each edge
resulting from a node reduction is $ab_{i}$ for a fork reduction or $b_{i}a$
for a joint reduction. Node reductions are used until series and parallel
reductions are possible. Ultimately, the single edge to which $G$ is reduced
has a label giving an expression of $G$.

In relation to a Fibonacci graph that has more than $3$ vertices, this
algorithm is transformed to the following special procedure.

\begin{enumerate}
\item A fork reduction is done at the second vertex (from the left) (see
Figure \ref{fig2}) or a joint reduction is done at the last but one vertex
(from the left).

\item A parallel reduction is done at the first and the second vertices
(from the left) in the case of preceding fork reduction or at the last but
one and the last vertices (from the left) in the case of preceding joint
reduction.

\item If the resulting reduced $FG$ contains more than three vertices, then
we return to step $1$ of this algorithm. Otherwise, it is a series-parallel
graph which is reduced to a single edge by a series reduction at the second
vertex and a parallel reduction at the source and the sink. The single edge
is labeled by the resulting expression.
\end{enumerate}

The example of the reduction process in relation to a $6$-vertex $FG$ is
shown in Figure \ref{fg_fig4}. 
\begin{figure}[tbp]
\setlength{\unitlength}{1.0cm}
\par
\begin{picture}(5,2)(-7,-0.5)\thicklines

\put(-4,0){\makebox(0,0){Initial Fibonacci graph}}

\multiput(0,0)(1.5,0){6}{\circle*{0.15}}

\put(0,-0.3){\makebox(0,0){1}} \put(1.5,-0.3){\makebox(0,0){2}}
\put(3,-0.3){\makebox(0,0){3}} \put(4.5,-0.3){\makebox(0,0){4}}
\put(6,-0.3){\makebox(0,0){5}} \put(7.5,-0.3){\makebox(0,0){6}}

\multiput(0,0)(1.5,0){5}{\vector(1,0){1.5}}

\put(0.75,-0.2){\makebox(0,0){$a_{1}$}}
\put(2.25,-0.2){\makebox(0,0){$a_{2}$}}
\put(3.75,-0.2){\makebox(0,0){$a_{3}$}}
\put(5.25,-0.2){\makebox(0,0){$a_{4}$}}
\put(6.75,-0.2){\makebox(0,0){$a_{5}$}}

\qbezier(0,0)(1.5,2)(3,0) \qbezier(1.5,0)(3,2)(4.5,0)
\qbezier(3,0)(4.5,2)(6,0) \qbezier(4.5,0)(6,2)(7.5,0)

\multiput(3.085,0)(1.5,0){4}{\vector(3,-2){0}}

\put(1.5,1.3){\makebox(0,0){$b_{1}$}}
\put(3,1.3){\makebox(0,0){$b_{2}$}}
\put(4.5,1.3){\makebox(0,0){$b_{3}$}}
\put(6,1.3){\makebox(0,0){$b_{4}$}}

\end{picture}
\par
\begin{picture}(5,2)(-7,0.5)\thicklines

\put(-4,1.5){\makebox(0,0){Joint reduction at vertex 5}}

\multiput(0,0)(1.5,0){5}{\circle*{0.15}}

\put(0,-0.3){\makebox(0,0){1}} \put(1.5,-0.3){\makebox(0,0){2}}
\put(3,-0.3){\makebox(0,0){3}} \put(4.5,-0.3){\makebox(0,0){4}}
\put(6,-0.3){\makebox(0,0){6}}

\multiput(0,0)(1.5,0){4}{\vector(1,0){1.5}}

\put(0.75,-0.2){\makebox(0,0){$a_{1}$}}
\put(2.25,-0.2){\makebox(0,0){$a_{2}$}}
\put(3.75,-0.2){\makebox(0,0){$a_{3}$}}
\put(5.25,-0.2){\makebox(0,0){$a_{4}a_{5}$}}

\qbezier(0,0)(1.5,2)(3,0) \qbezier(1.5,0)(3,2)(4.5,0)
\qbezier(3,0)(4.5,2)(6,0) \qbezier(4.5,0)(5.25,1)(6,0)

\multiput(3.085,0)(1.5,0){3}{\vector(3,-2){0}}

\put(1.5,1.3){\makebox(0,0){$b_{1}$}}
\put(3,1.3){\makebox(0,0){$b_{2}$}}
\put(4.5,1.3){\makebox(0,0){$b_{3}a_{5}$}}
\put(5.25,0.75){\makebox(-1,-0.2){$b_{4}$}}

\end{picture}
\par
\begin{picture}(5,2)(-7,1.5)\thicklines

\put(-4,1.5){\makebox(0,0){Parallel reduction at vertices 4 and 6}}

\multiput(0,0)(1.5,0){5}{\circle*{0.15}}

\put(0,-0.3){\makebox(0,0){1}} \put(1.5,-0.3){\makebox(0,0){2}}
\put(3,-0.3){\makebox(0,0){3}} \put(4.5,-0.3){\makebox(0,0){4}}
\put(6,-0.3){\makebox(0,0){6}}

\multiput(0,0)(1.5,0){4}{\vector(1,0){1.5}}

\put(0.75,-0.2){\makebox(0,0){$a_{1}$}}
\put(2.25,-0.2){\makebox(0,0){$a_{2}$}}
\put(3.75,-0.2){\makebox(0,0){$a_{3}$}}
\put(5.25,-0.9){\makebox(0,0){$b_{4}+a_{4}a_{5}$}}

\thinlines \put(5,-0.7){\line(1,2){0.3}} \thicklines

\qbezier(0,0)(1.5,2)(3,0) \qbezier(1.5,0)(3,2)(4.5,0)
\qbezier(3,0)(4.5,2)(6,0)

\multiput(3.085,0)(1.5,0){3}{\vector(3,-2){0}}

\put(1.5,1.3){\makebox(0,0){$b_{1}$}}
\put(3,1.3){\makebox(0,0){$b_{2}$}}
\put(4.5,1.3){\makebox(0,0){$b_{3}a_{5}$}}

\end{picture}
\par
\begin{picture}(5,2)(-7,2.5)\thicklines

\put(-4,1.7){\makebox(0,0){Fork reduction at vertex 2}}
\put(-4,1.2){\makebox(0,0){Parallel reduction at vertices 1 and 3}}

\multiput(0,0)(1.5,0){4}{\circle*{0.15}}

\put(0,-0.3){\makebox(0,0){1}} \put(1.5,-0.3){\makebox(0,0){3}}
\put(3,-0.3){\makebox(0,0){4}} \put(4.5,-0.3){\makebox(0,0){6}}

\multiput(0,0)(1.5,0){3}{\vector(1,0){1.5}}

\put(0.75,-0.9){\makebox(0,0){$b_{1}+a_{1}a_{2}$}}
\put(2.25,-0.2){\makebox(0,0){$a_{3}$}}
\put(3.75,-0.9){\makebox(0,0){$b_{4}+a_{4}a_{5}$}}

\thinlines \put(0.5,-0.7){\line(1,2){0.3}}
\put(3.5,-0.7){\line(1,2){0.3}} \thicklines

\qbezier(0,0)(1.5,2)(3,0) \qbezier(1.5,0)(3,2)(4.5,0)

\multiput(3.085,0)(1.5,0){2}{\vector(3,-2){0}}

\put(1.5,1.3){\makebox(0,0){$a_{1}b_{2}$}}
\put(3,1.3){\makebox(0,0){$b_{3}a_{5}$}}

\end{picture}
\par
\begin{picture}(5,2)(-7,4)\thicklines

\put(-4,1.7){\makebox(0,0){Fork reduction at vertex 3}}
\put(-4,1.2){\makebox(0,0){Parallel reduction at vertices 1 and 4}}

\multiput(0,0)(2.25,0){3}{\circle*{0.15}}

\put(0,-0.3){\makebox(0,0){1}} \put(2.25,-0.3){\makebox(0,0){4}}
\put(4.5,-0.3){\makebox(0,0){6}}

\multiput(0,0)(2.25,0){2}{\vector(1,0){2.25}}

\put(0.2,-0.9){\makebox(0,0){$a_{1}b_{2}+(b_{1}+a_{1}a_{2})a_{3}$}}
\put(3.8,-0.9){\makebox(0,0){$b_{4}+a_{4}a_{5}$}}

\thinlines \put(0.5,-0.7){\line(1,2){0.3}}
\put(3.5,-0.7){\line(1,2){0.3}} \thicklines

\qbezier(0,0)(2.25,2)(4.5,0)

\multiput(4.585,0)(1.5,0){1}{\vector(3,-2){0}}

\put(2.25,1.3){\makebox(0,0){$(b_{1}+a_{1}a_{2})b_{3}a_{5}$}}

\end{picture}
\par
\begin{picture}(5,2)(-7,5.5)\thicklines

\put(-4,1.7){\makebox(0,0){Series reduction at vertex 4}}
\put(-4,1.2){\makebox(0,0){Parallel reduction at vertices 1 and 6}}

\multiput(0,0)(4.5,0){2}{\circle*{0.15}}

\put(0,-0.3){\makebox(0,0){1}} \put(4.5,-0.3){\makebox(0,0){6}}

\multiput(0,0)(4.5,0){1}{\vector(1,0){4.5}}

\put(2.25,0.4){\makebox(0,0){$(b_{1}+a_{1}a_{2})b_{3}a_{5}+(a_{1}b_{2}+(b_{1}+a_{1}a_{2})a_{3})(b_{4}+a_{4}a_{5}$)}}

\end{picture}
\par
\begin{picture}(5,6)(-7,0.5)\thicklines
\end{picture}
\caption{The example of a reduction process on a Fibonacci graph.}
\label{fg_fig4}
\end{figure}
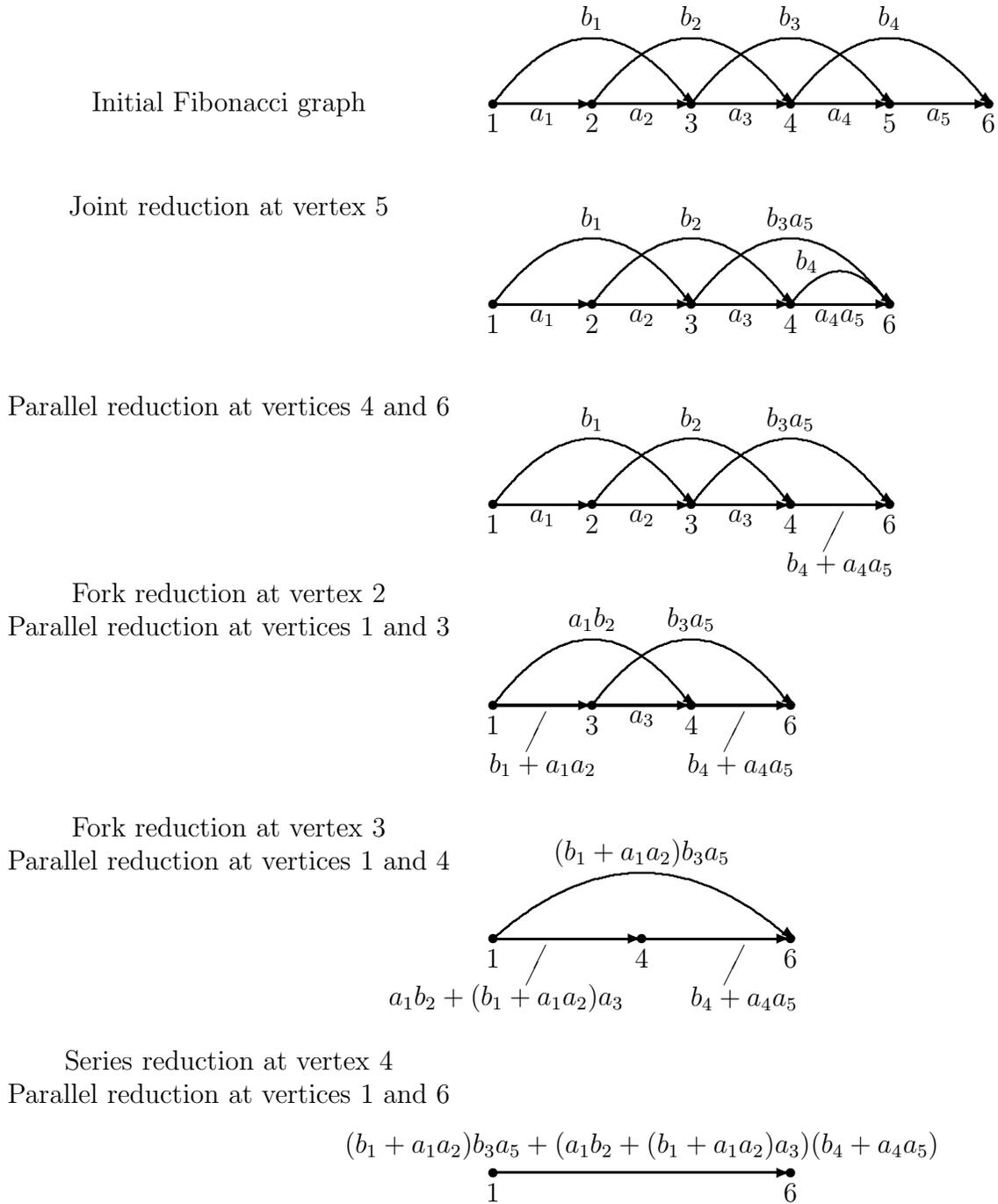

It is clear (see Figure \ref{fg_fig4}) that if the $FG$ contains $n$
vertices, then the number of applied node reductions is equal to $y=n-3$.
Thus, the reduction method applied to an $n$-vertex $FG$ includes $2^{n-3}$
possible reduction processes that are due to different numbers and execution
orders of fork and joint reductions. In the case of $n\leq 3$ the $n$-vertex 
$FG$ is a series-parallel graph and node reductions are not done. Our
intention is to find reduction processes which lead to expression
representation $Ex(FG)$ with a minimum complexity. We propose the following
algorithm:

\begin{enumerate}
\item $y\leftarrow n-3$

\item \label{r2}$\mathbf{if}$ $y$ $\func{mod}$ $2$ $=$ $0$

\item \qquad \label{r3}$fork\_count\leftarrow y/2\qquad
joint\_count\leftarrow y/2$

\item $\mathbf{else}$

\item $\qquad $\label{r5}$z\leftarrow \mathbf{rand}(2)$

\item \qquad $\mathbf{if}$ $z=1$

\item \qquad \qquad $fork\_count\leftarrow \left( y-1\right) /2\qquad
joint\_count\leftarrow \left( y+1\right) /2$

\item \qquad $\mathbf{else}$

\item \qquad \qquad \label{r9}$fork\_count\leftarrow \left( y+1\right)
/2\qquad joint\_count\leftarrow \left( y-1\right) /2$\qquad

\item \label{r10}$\mathbf{while}$ $fork\_count$ $>0$ $and$ $joint\_count$ $%
>0 $

\item \qquad $z\leftarrow \mathbf{rand}(2)$

\item \qquad $\mathbf{if}$ $z=1$

\item $\qquad \qquad \mathbf{apply}$\textbf{\ }$\mathbf{fork}$\textbf{\ }$%
\mathbf{and}$\textbf{\ }$\mathbf{parallel}$\textbf{\ }$\mathbf{reductions}$

\item \qquad \qquad $fork\_count\leftarrow fork\_count-1$

\item \qquad $\mathbf{else}$

\item $\qquad \qquad \mathbf{apply}$\textbf{\ }$\mathbf{joint}$\textbf{\ }$%
\mathbf{and}$\textbf{\ }$\mathbf{parallel}$\textbf{\ }$\mathbf{reductions}$

\item \label{r17}\qquad \qquad $joint\_count\leftarrow joint\_count-1$

\item \label{r18}$\mathbf{while}$ $fork\_count$ $>0$

\item $\qquad \mathbf{apply}$\textbf{\ }$\mathbf{fork}$\textbf{\ }$\mathbf{%
and}$\textbf{\ }$\mathbf{parallel}$\textbf{\ }$\mathbf{reductions}$

\item \label{r20}\qquad $fork\_count\leftarrow fork\_count-1$

\item \label{r21}$\mathbf{while}$ $joint\_count$ $>0$

\item $\qquad \mathbf{apply}$\textbf{\ }$\mathbf{joint}$\textbf{\ }$\mathbf{%
and}$\textbf{\ }$\mathbf{parallel}$\textbf{\ }$\mathbf{reductions}$

\item \label{r23}\qquad $joint\_count\leftarrow joint\_count-1$
\end{enumerate}

Lines \ref{r2} and \ref{r3} of the algorithm determine the numbers of fork
and joint reductions to be applied for even $y$ (odd $n$). For odd $y$ (even 
$n$) there are two possible values for the numbers of fork and joint
reductions. These values are determined in lines \ref{r5}-\ref{r9} by
function $\mathbf{rand}(2)$ which generates randomly $1$ or $2$. The $%
\mathbf{while}$ loop in lines \ref{r10}-\ref{r17} repeatedly applies a pair
of node and parallel reductions. A kind of a node reduction (fork or joint)
is also determined by function $\mathbf{rand}(2)$. After all possible joint
or all possible fork reductions are done, the remaining fork (lines \ref{r18}%
-\ref{r20}) or remaining joint (lines \ref{r21}-\ref{r23}) reductions,
respectively, are applied.

We define this algorithm as the \textit{optimal reduction method}. $\qquad
\qquad \quad ~$

\begin{theorem}
\label{th_fg_red_min}The minimum complexity representation (for terms and
for plus operators) among all possible expression representations $Ex(FG)$
derived by the reduction method for an $n$-vertex $FG$ ($n>3$) is achieved
by the optimal reduction method.
\end{theorem}

\begin{proof}
Actually, as follows from lines \ref{r2}-\ref{r9} of the optimal
decomposition method, we should prove that the minimum complexity
representation (for terms and for plus operators) is achieved if and only if
(i) the number of applied fork reductions is equal to the number of applied
joint reductions for odd $n$; (ii) the numbers of fork reductions and joint
reductions are distinguished by one for even $n$.

The initial number of terms on the graph edges is equal to the number of
edges. The initial number of plus operators is equal to $0$. Each pair of
node and parallel reductions leads to an increase in the total number of
terms and plus operators on the graph edges. This follows from the fact that
each node reduction leads to duplicate copies of corresponding terms and
subexpressions and to new plus operators. We should find the numbers of
terms and plus operators on the single edge to which the $FG$ is reduced and
analyze how these numbers increase in comparison to with their initial
values. The following basic points are used.\smallskip

\textbf{1.} The increment value of the total number of terms and plus
operators on the left side of the $FG$ increases as the number of fork
reductions increases; the increment value on the right side of the $FG$
increases as the number of joint reductions increases.\smallskip

Indeed, the current increment depends only on the duplicate subexpression
that labels the appointed edge. This edge is positioned between the first
and the second vertices in the case of a fork reduction, or between the last
but one and the last vertices in the case of a joint reduction. The
increment of the first complexity characteristic is equal to the total
number of terms in the duplicate subexpression. The increment of the second
complexity characteristic is equal to the number of plus operators in this
subexpression with an additional plus operator added. In the next step of
the reduction algorithm, the above mentioned edge is the result of parallel
reduction at two edges. One of these edges is labeled by the subexpression
including the duplicate subexpression of the previous step (see Figure \ref%
{fg_fig4}). Hence, the size of the duplicate subexpression for the next
reduction increases and, therefore, the increment value increases.\smallskip

\textbf{2.} The above mentioned increment values from the left side and from
the right side of the $FG$ are independent, \thinspace i.e., \thinspace the
\thinspace increment \thinspace value \thinspace after \thinspace a
\thinspace current \thinspace fork \thinspace reduction \thinspace does
\thinspace not \thinspace depend \thinspace on \thinspace the \thinspace
number \thinspace of \thinspace already \thinspace applied \thinspace joint
\thinspace reductions \thinspace and \thinspace vice \thinspace
versa.\smallskip

We prove the stronger statement which asserts that the above mentioned
duplicate subexpressions determining the increments are independent. In
every reduction step, we conditionally denote the edge labels of the reduced 
$FG$ as the edge labels of the initial $FG$ shown in Figure \ref{fig2}. In
such a case, in each reduction step, edges leaving the source of the reduced 
$FG$ are labeled $a_{1}$ and $b_{1}$, respectively, and edges entering the
sink of the reduced $FG$ are labeled $a_{n-1}$ and $b_{n-2}$, respectively.
Hence, $a_{1}$ is a current duplicate subexpression before a fork reduction
and $a_{n-1}$ is a current duplicate subexpression before a joint reduction.
The new (after a fork and parallel reduction) $a_{1}$ depends only on the
old (before a fork reduction) $a_{1}$, $a_{2}$, and $b_{1}$. The new (after
a joint and parallel reduction) $a_{n-1}$ depends only on the old (before a
joint reduction) $a_{n-1}$, $a_{n-2}$, and $b_{n-2}$. Labels $a_{2}$ and $%
a_{n-2}$ are initial terms always. They do not change and depend on nothing.
The new $b_{1}$ depends only on the old $a_{1}$ and $b_{2}$. The new $%
b_{n-2} $ depends only on the old $a_{n-1}$ and $b_{n-3}$. Labels $b_{2}$
and $b_{n-3}$ are also initial terms always, and depend on nothing. For this
reason, the new $a_{1}$ in no step depends on the old $a_{n-1}$ and $b_{n-2}$
and the new $a_{n-1}$ in no step depends on the old $a_{1}$ and $b_{1}$. On
the other hand, the $Ex(FG)$ accumulation takes place only in $a_{1}$ and $%
b_{1}$, and in $a_{n-1}$ and $b_{n-2}$. The pair of the fork and parallel
reductions influences only the new $a_{1}$ and $b_{1}$, and the pair of the
joint and parallel reductions influences only the new $a_{n-1}$ and $b_{n-2}$%
. All the above holds for all reduction steps including the last one, when $%
a_{1}$ and $a_{n-1}$ draw closer to one another. Figure \ref{fg_fig5}, where
the reduction method is applied to a $5$-vertex $FG$, illustrates this
phenomenon. Two possible algorithms in which either the fork reduction comes
first (and the joint reduction follows) or the joint reduction comes first
(and the fork reduction follows) are illustrated. As shown, both algorithms
lead to the same result. Labels on edges $(1,3)$ and $(3,5)$ of the
resulting $3$-edge st-dag have no common terms. That is, $a_{1}$ does not
depend on joint reductions and $a_{n-1}$ does not depend on fork reductions.
Therefore, duplicate subexpressions from the left side and from the
right\thinspace side are \thinspace independent, and,\thinspace thus,
corresponding increment values are independent as well.

\begin{figure}[tbph]
\setlength{\unitlength}{1.0cm}
\par
\begin{picture}(5,2)(-4,-0.5)\thicklines

\multiput(0,0)(1.5,0){5}{\circle*{0.15}}

\put(0,-0.3){\makebox(0,0){1}} \put(1.5,-0.3){\makebox(0,0){2}}
\put(3,-0.3){\makebox(0,0){3}} \put(4.5,-0.3){\makebox(0,0){4}}
\put(6,-0.3){\makebox(0,0){5}}

\multiput(0,0)(1.5,0){4}{\vector(1,0){1.5}}

\put(0.75,-0.2){\makebox(0,0){$a_{1}$}}
\put(2.25,-0.2){\makebox(0,0){$a_{2}$}}
\put(3.75,-0.2){\makebox(0,0){$a_{3}$}}
\put(5.25,-0.2){\makebox(0,0){$a_{4}$}}

\qbezier(0,0)(1.5,2)(3,0) \qbezier(1.5,0)(3,2)(4.5,0)
\qbezier(3,0)(4.5,2)(6,0)

\multiput(3.085,0)(1.5,0){3}{\vector(3,-2){0}}

\put(1.5,1.3){\makebox(0,0){$b_{1}$}}
\put(3,1.3){\makebox(0,0){$b_{2}$}}
\put(4.5,1.3){\makebox(0,0){$b_{3}$}}

\put(1.5,-1.2){\vector(-1,-1){1.5}}
\put(4.5,-1.2){\vector(1,-1){1.5}}

\put(-2,-1.3){\makebox(0,0){Fork red. at vertex 2}}
\put(-2,-1.8){\makebox(0,0){Par. red. at vertices 1 and 3}}
\put(8,-1.3){\makebox(0,0){Joint red. at vertex 4}}
\put(8,-1.8){\makebox(0,0){Par. red. at vertices 3 and 5}}

\end{picture}
\par
\begin{picture}(5,2)(-8.5,2.5)\thicklines

\multiput(0,0)(1.5,0){4}{\circle*{0.15}}

\put(0,-0.3){\makebox(0,0){1}} \put(1.5,-0.3){\makebox(0,0){2}}
\put(3,-0.3){\makebox(0,0){3}} \put(4.5,-0.3){\makebox(0,0){5}}

\multiput(0,0)(1.5,0){3}{\vector(1,0){1.5}}

\put(0.75,-0.2){\makebox(0,0){$a_{1}$}}
\put(2.25,-0.2){\makebox(0,0){$a_{2}$}}
\put(3.75,-0.9){\makebox(0,0){$b_{3}+a_{3}a_{4}$}}

\thinlines \put(3.5,-0.7){\line(1,2){0.3}} \thicklines

\qbezier(0,0)(1.5,2)(3,0) \qbezier(1.5,0)(3,2)(4.5,0)

\multiput(3.085,0)(1.5,0){2}{\vector(3,-2){0}}

\put(1.5,1.3){\makebox(0,0){$b_{1}$}}
\put(3,1.3){\makebox(0,0){$b_{2}a_{4}$}}

\put(1.5,-2.2){\vector(-1,-1){1.5}}

\put(3.5,-2.6){\makebox(0,0){Fork red. at vertex 2}}
\put(3.5,-3.1){\makebox(0,0){Par. red. at vertices 1 and 3}}

\end{picture}
\par
\begin{picture}(5,2)(-1,0.5)\thicklines

\multiput(0,0)(1.5,0){4}{\circle*{0.15}}

\put(0,-0.3){\makebox(0,0){1}} \put(1.5,-0.3){\makebox(0,0){3}}
\put(3,-0.3){\makebox(0,0){4}} \put(4.5,-0.3){\makebox(0,0){5}}

\multiput(0,0)(1.5,0){3}{\vector(1,0){1.5}}

\put(0.75,-0.9){\makebox(0,0){$b_{1}+a_{1}a_{2}$}}
\put(2.25,-0.2){\makebox(0,0){$a_{3}$}}
\put(3.75,-0.2){\makebox(0,0){$a_{4}$}}

\thinlines \put(0.5,-0.7){\line(1,2){0.3}} \thicklines

\qbezier(0,0)(1.5,2)(3,0) \qbezier(1.5,0)(3,2)(4.5,0)

\multiput(3.085,0)(1.5,0){2}{\vector(3,-2){0}}

\put(1.5,1.3){\makebox(0,0){$a_{1}b_{2}$}}
\put(3,1.3){\makebox(0,0){$b_{3}$}}

\put(3,-2.2){\vector(1,-1){1.5}}

\put(1,-2.6){\makebox(0,0){Joint red. at vertex 4}}
\put(1,-3.1){\makebox(0,0){Par. red. at vertices 3 and 5}}

\end{picture}
\par
\begin{picture}(5,2)(-4.75,4)\thicklines

\multiput(0,0)(2.25,0){3}{\circle*{0.15}}

\put(0,-0.3){\makebox(0,0){1}} \put(2.25,-0.3){\makebox(0,0){3}}
\put(4.5,-0.3){\makebox(0,0){5}}

\multiput(0,0)(2.25,0){2}{\vector(1,0){2.25}}

\put(1.125,-0.2){\makebox(0,0){$b_{1}+a_{1}a_{2}$}}
\put(3.375,-0.2){\makebox(0,0){$b_{3}+a_{3}a_{4}$}}

\qbezier(0,0)(2.25,2)(4.5,0)

\multiput(4.585,0)(1.5,0){1}{\vector(3,-2){0}}

\put(2.25,1.3){\makebox(0,0){$a_{1}b_{2}a_{4}$}}

\end{picture}
\par
\begin{picture}(5,4.5)(-7,0.5)\thicklines
\end{picture}
\caption{Two reduction algorithms on a $5$-vertex Fibonacci graph leading to
the same result.}
\label{fg_fig5}
\end{figure}
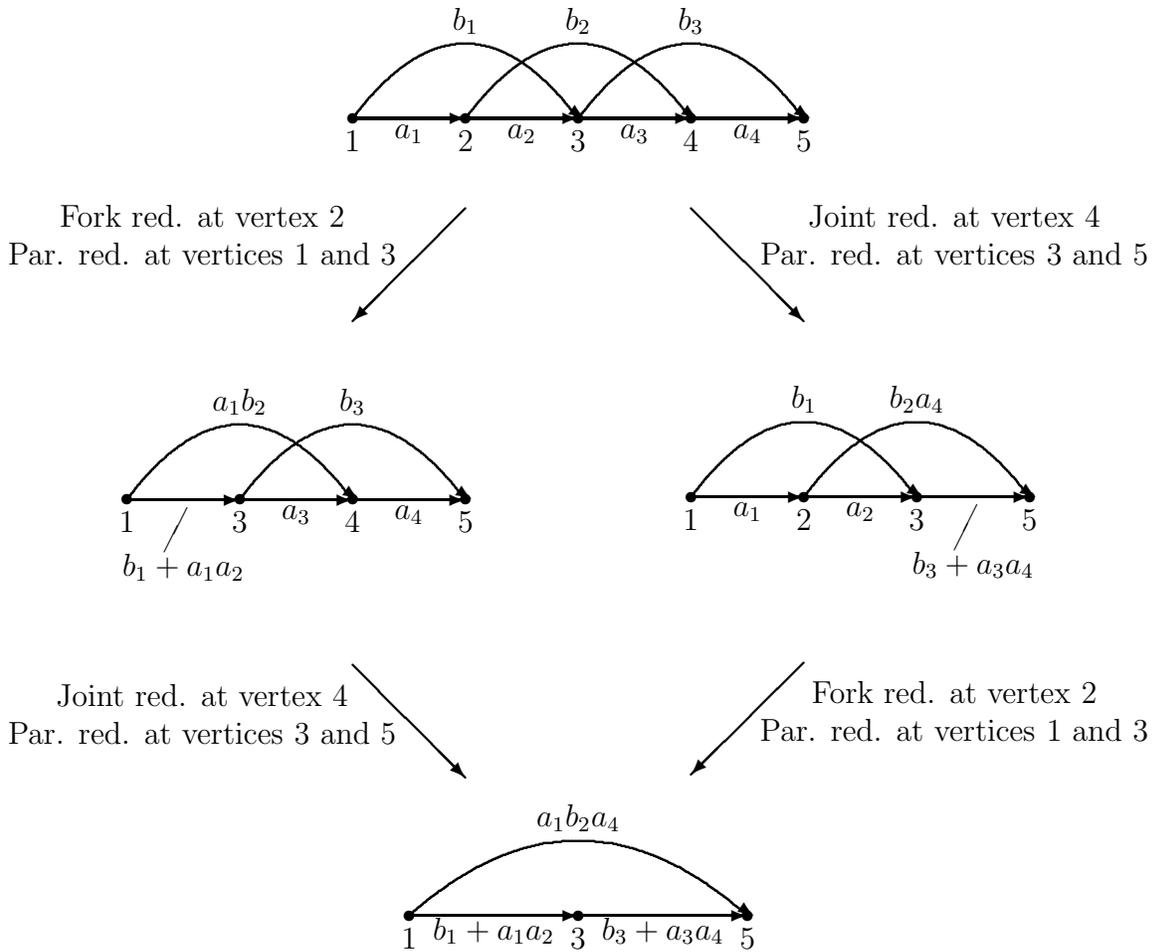

\textbf{3.} The increment value related to the $i$-th fork reduction in the $%
FG$ is equal to the increment value related to the $i$-th joint reduction in
the $FG$.\smallskip

This follows from the symmetrical structure of an $FG$ and from independence
of fork and joint reductions.\smallskip

As noted above, if the $FG$ contains $n$ vertices, then the number of
applied node reductions is equal to $y=n-3$. Since in each step, two kinds
of node reductions are possible, initially, the potential number of possible
node reductions is equal to $2y$ ($y$ fork and $y$ joint reductions). Hence,
we can present the reduction procedure as follows. There are two equal
stacks $S_{1}$ and $S_{2}$ (Figure \ref{fg_fig6}). Each of them contains $y$
elements. The size of an element in each stack increases from top to bottom.
The elements of the same level in $S_{1}$ and $S_{2}$ are of equal size.
Here a $Pop$ operation on $S_{1}$ corresponds to a fork reduction and a $Pop$
operation on $S_{2}$ corresponds to a joint reduction. The size of a stack
element corresponds to an increment value. We should put out $y$ elements
from two stacks. It is clear that for even $y$ (odd $n$) the total size of
pulled out elements will be minimum if and only if a $Pop$ operation is done 
$y/2$ times on each stack. For odd $y$ (even $n$), in order to ensure the
minimum total size, a $Pop$ operation should be done $(y-1)/2$ times on $%
S_{1}$ and $(y-1)/2+1$ times on $S_{2}$, or vice versa. Thus, the proof of
the theorem is complete. 
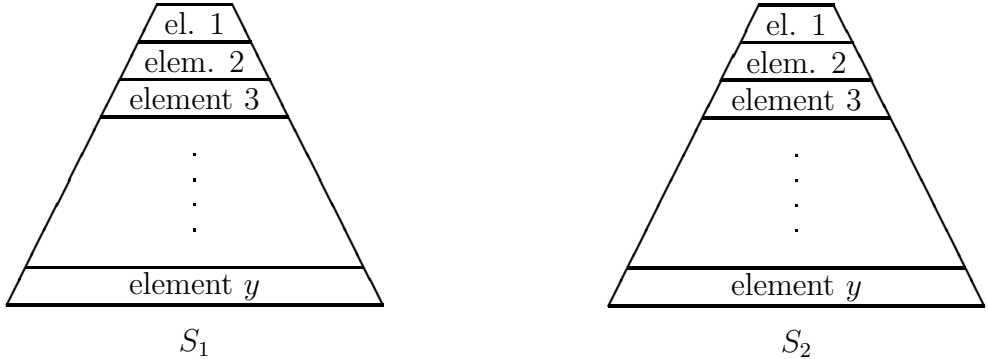
\begin{figure}[tbp]
\setlength{\unitlength}{1.0cm}
\par
\begin{picture}(12,5)(0,0.8)\thicklines

\drawline(2,5)(0,1)(5,1)(3,5)(2,5) \drawline(1.75,4.5)(3.25,4.5)
\drawline(1.5,4)(3.5,4) \drawline(1.25,3.5)(3.75,3.5)
\dottedline{0.3}(2.5,3)(2.5,2) \drawline(0.25,1.5)(4.75,1.5)

\put(2.5,4.75){\makebox(0,0){el. 1}}
\put(2.5,4.25){\makebox(0,0){elem. 2}}
\put(2.5,3.75){\makebox(0,0){element 3}}
\put(2.5,1.25){\makebox(0,0){element $y$}}

\put(2.5,0.5){\makebox(0,0){$S_{1}$}}

\end{picture}
\par
\begin{picture}(12,0)(-8,0.3)\thicklines

\drawline(2,5)(0,1)(5,1)(3,5)(2,5) \drawline(1.75,4.5)(3.25,4.5)
\drawline(1.5,4)(3.5,4) \drawline(1.25,3.5)(3.75,3.5)
\dottedline{0.3}(2.5,3)(2.5,2) \drawline(0.25,1.5)(4.75,1.5)

\put(2.5,4.75){\makebox(0,0){el. 1}}
\put(2.5,4.25){\makebox(0,0){elem. 2}}
\put(2.5,3.75){\makebox(0,0){element 3}}
\put(2.5,1.25){\makebox(0,0){element $y$}}

\put(2.5,0.5){\makebox(0,0){$S_{2}$}}

\end{picture}
\caption{Two equal stacks. Sizes of their elements correspond to increment
values.}
\label{fg_fig6}
\end{figure}
\end{proof}

\begin{proposition}
The expression representation $Ex(FG)$ derived by the reduction method does
not depend on the execution order of fork and joint reductions and depends
only on their number.
\end{proposition}

\begin{proof}
The proof is similar to point 2 in the proof of Theorem \ref{th_fg_red_min}.
As noted, the $Ex(FG)$ accumulation takes place only in $a_{1}$ and $b_{1}$,
and in $a_{n-1}$ and $b_{n-2}$, i.e., on edges leaving the source and
entering the sink of the current reduced $FG$. Accumulation processes on
these pairs of edges are independent. As shown in Figure \ref{fg_fig5}, in
the last reduction steps, when the source and the sink draw closer to one
another, edge labels of the resulting $3$-edge st-dag do not depend also on
the order of the execution of reductions.\medskip
\end{proof}

\begin{remark}
The number of possible reduction processes included by the optimal reduction
method applied to an $n$-vertex $FG$ is equal to%
\begin{equation*}
\binom{n-3}{\left( n-3\right) /2}=\frac{\left( n-3\right) !}{\left( \left(
\left( n-3\right) /2\right) !\right) ^{2}}
\end{equation*}%
for odd $n$ and 
\begin{equation*}
\binom{n-3}{\left( n-4\right) /2}=\frac{\left( n-3\right) !}{\left( \left(
n-4\right) /2\right) !\left( \left( n-2\right) /2\right) !}
\end{equation*}%
for even $n$.

The different reduction processes are due to different execution orders of
fork and joint reductions.\medskip
\end{remark}

It may be noted that each pair of a fork and a parallel reduction
corresponds to two parallel recursion steps of the DFS method, each of which
is executed by traversing the $FG$ in opposite directions. These steps are
equivalent to an ordinary $Ex(FG)$ accumulation step on two edges, leaving
the source of the current reduced $FG$. By analogy, each pair of a joint and
a parallel reduction corresponds to two parallel recursion steps of the
direct DFS method. They are equivalent to an ordinary $Ex(FG)$ accumulation
step on two edges, entering the sink of the current reduced $FG$. Hence, the
reduction process in an $FG$ can be conditionally presented as four parallel
DFS processes on four subgraphs of this $FG$. The corresponding four
subexpressions are linked in the final step of the reduction procedure (see
Figure \ref{fg_fig4}). The resulting expression is constructed from the
following three elements: four subexpressions which are obtained by applying
the DFS method to four corresponding subgraphs of the $FG$; one additional
term $b_{l}$ ($l$ determines the place where the $FG$ is decomposed into
subgraphs, and depends on the number of fork and joint reductions); and one
additional plus operation.

If only fork or only joint reductions are applied to a Fibonacci graph, then
the reduction method gives the DFS method. Hence, the DFS method is a
special (worst from the perspective of the complexity) case of the reduction
method.

The optimal reduction method is equivalent to applying the DFS method to
four subgraphs which are revealed by decomposing the $FG$ in the middle.
This results in the following correlations between the reduction method and
the DFS method: 
\begin{eqnarray}
T_{r}(n) &=&T_{DFS}\left( \left\lceil \frac{n}{2}\right\rceil \right)
+T_{DFS}\left( \left\lfloor \frac{n}{2}\right\rfloor +1\right) +  \notag \\
&&T_{DFS}\left( \left\lceil \frac{n}{2}\right\rceil -1\right) +T_{DFS}\left(
\left\lfloor \frac{n}{2}\right\rfloor \right) +1\text{ }  \label{fgf5}
\end{eqnarray}%
\begin{eqnarray}
P_{r}(n) &=&P_{DFS}\left( \left\lceil \frac{n}{2}\right\rceil \right)
+P_{DFS}\left( \left\lfloor \frac{n}{2}\right\rfloor +1\right) +  \notag \\
&&P_{DFS}\left( \left\lceil \frac{n}{2}\right\rceil -1\right) +P_{DFS}\left(
\left\lfloor \frac{n}{2}\right\rfloor \right) +1.  \label{fgf6}
\end{eqnarray}%
Here $T_{r}(n)$ and $T_{DFS}(n)$ are the first complexity characteristics of 
$Ex(FG)$ which is derived by the optimal reduction method and the DFS
method, respectively, for an $n$-vertex $FG$; $P_{r}(n)$ and $P_{DFS}(n)$
are the second complexity characteristics of $Ex(FG)$ which is derived by
the optimal reduction method and the DFS methods, respectively, for an $n$%
-vertex $FG$.

\begin{theorem}
\label{th_red_exp}For an $n$-vertex $FG$:

1. The total number of terms $T(n)$ in the expression $Ex(FG)$ derived by
the optimal reduction method is defined recursively as follows: 
\begin{eqnarray*}
T(1) &=&0 \\
T(2) &=&1 \\
T(3) &=&3 \\
T(4) &=&6 \\
T(5) &=&9 \\
T(n) &=&T(n-2)+T(n-4)+7\quad (n>5).
\end{eqnarray*}%
2. The number of plus operators $P(n)$ in the expression $Ex(FG)$ derived by
the optimal reduction method is defined recursively as follows: 
\begin{eqnarray*}
P(1) &=&0 \\
P(2) &=&0 \\
P(3) &=&1 \\
P(4) &=&2 \\
P(5) &=&3 \\
P(n) &=&P(n-2)+P(n-4)+3\quad (n>5).
\end{eqnarray*}
\end{theorem}

\begin{proof}
1. Initial statements $T(1)=0$, $T(2)=1$, $T(3)=3$, $T(4)=6$, $T(5)=9$ can
be checked. For $n>5$, we consider odd and even $n$ and use formulae (\ref%
{fgf3}) and (\ref{fgf5}).

(i) Odd $n$%
\begin{eqnarray*}
T_{r}(n) &=&2T_{DFS}\left( \frac{n-1}{2}+1\right) +2T_{DFS}\left( \frac{n-1}{%
2}\right) +1 \\
&=&2\left( T_{DFS}\left( \frac{n-1}{2}\right) +T_{DFS}\left( \frac{n-1}{2}%
-1\right) +2\right) + \\
&&2\left( T_{DFS}\left( \frac{n-1}{2}-1\right) +T_{DFS}\left( \frac{n-1}{2}%
-2\right) +2\right) +1 \\
&=&2T_{DFS}\left( \frac{n-1}{2}\right) +2T_{DFS}\left( \frac{n-1}{2}%
-1\right) + \\
&&2T_{DFS}\left( \frac{n-1}{2}-1\right) +2T_{DFS}\left( \frac{n-1}{2}%
-2\right) +9 \\
&=&T_{r}(n-2)-1+T_{r}(n-4)-1+9=T_{r}(n-2)+T_{r}(n-4)+7.
\end{eqnarray*}%
(ii) Even $n$%
\begin{eqnarray*}
T_{r}(n) &=&T_{DFS}\left( \frac{n}{2}+1\right) +2T_{DFS}\left( \frac{n}{2}%
\right) +T_{DFS}\left( \frac{n}{2}-1\right) +1 \\
&=&T_{DFS}\left( \frac{n}{2}\right) +T_{DFS}\left( \frac{n}{2}-1\right) +2+
\\
&&2\left( T_{DFS}\left( \frac{n}{2}-1\right) +T_{DFS}\left( \frac{n}{2}%
-2\right) +2\right) + \\
&&T_{DFS}\left( \frac{n}{2}-2\right) +T_{DFS}\left( \frac{n}{2}-3\right) +2+1
\\
&=&T_{r}(n-2)-1+T_{r}(n-4)-1+9=T_{r}(n-2)+T_{r}(n-4)+7.
\end{eqnarray*}%
2. The proof is analogous and is based on formulae (\ref{fgf4}) and (\ref%
{fgf6}).\medskip
\end{proof}

As follows from Theorems \ref{th_seq_exp}, \ref{th_dfs_exp}, \ref{th_dls_exp}%
, and \ref{th_red_exp}, the optimal decomposition method is more efficient
than all the methods presented in section \ref{simple} from the perspective
of both complexity characteristics.

\begin{corollary}
\label{cor_red_exp}For an $n$-vertex $FG$:

1. The total number of terms $T(n)$ in the expression $Ex(FG)$ derived by
the optimal reduction method is expressed explicitly as follows: 
\begin{eqnarray*}
T(1) &=&0 \\
T(n) &\approx &\left( 4.\,\allowbreak 889\,6+\left( -1\right) ^{n}\cdot
0.07008\,9\right) \left( \sqrt{\frac{\sqrt{5}+1}{2}}\right) ^{n}+ \\
&&\left[ \left( 1+\left( -1\right) ^{n}\right) 0.02\,\allowbreak
016\,3+\left( 1+\left( -1\right) ^{n+1}\right) 0.08299\,6i\right] \left( i%
\sqrt{\frac{\sqrt{5}-1}{2}}\right) ^{n}-7\quad \\
(n &>&1)
\end{eqnarray*}

or%
\begin{eqnarray*}
T(1) &=&0 \\
\quad T(n) &\approx &\left( 4.\,\allowbreak 889\,6+\left( -1\right)
^{n}\cdot 0.07008\,9\right) \left( \sqrt{\frac{\sqrt{5}+1}{2}}\right) ^{n}+
\\
&&\left[ 0.04032\,5\cos \left( \frac{n\pi }{2}\right) -0.165\,99\sin \left( 
\frac{n\pi }{2}\right) \right] \left( \sqrt{\frac{\sqrt{5}-1}{2}}\right)
^{n}-7 \\
(n &>&1).
\end{eqnarray*}

2. The number of plus operators $P(n)$ in the expression $Ex(FG)$ derived by
the optimal reduction method is expressed explicitly as follows: 
\begin{eqnarray*}
P(1) &=&0 \\
P(n) &\approx &\left( 1.\,\allowbreak 867\,7+\left( -1\right) ^{n}\cdot
0.02\allowbreak 677\,2\right) \left( \sqrt{\frac{\sqrt{5}+1}{2}}\right) ^{n}+
\\
&&\left[ \left( 1+\left( -1\right) ^{n}\right) 0.05278\,6+\left( 1+\left(
-1\right) ^{n+1}\right) 0.217\,29i\right] \left( i\sqrt{\frac{\sqrt{5}-1}{2}}%
\right) ^{n}-3\quad \\
(n &>&1)
\end{eqnarray*}

or%
\begin{eqnarray*}
P(1) &=&0 \\
P(n) &\approx &\left( 1.\,\allowbreak 867\,7+\left( -1\right) ^{n}\cdot
0.02677\,2\right) \left( \sqrt{\frac{\sqrt{5}+1}{2}}\right) ^{n}+ \\
&&\left[ 0.105\,57\cos \left( \frac{n\pi }{2}\right) -0.434\,57\sin \left( 
\frac{n\pi }{2}\right) \right] \left( \sqrt{\frac{\sqrt{5}-1}{2}}\right)
^{n}-3 \\
(n &>&1).
\end{eqnarray*}
\end{corollary}

\bigskip The proof of Corollary \ref{cor_red_exp} uses the recurrences
obtained in Theorem \ref{th_red_exp} and is based on the method \cite{Ros}.

For $n=9$, the algebraic expression derived by the optimal reduction method
is 
\begin{eqnarray*}
&&(((a_{1}a_{2}+b_{1})a_{3}+a_{1}b_{2})a_{4}+(a_{1}a_{2}+b_{1})b_{3})(a_{5}(a_{6}(a_{7}a_{8}+b_{7})+b_{6}a_{8})+
\\
&&b_{5}(a_{7}a_{8}+b_{7}))+((a_{1}a_{2}+b_{1})a_{3}+a_{1}b_{2})b_{4}(a_{6}(a_{7}a_{8}+b_{7})+b_{6}a_{8}).
\end{eqnarray*}%
It contains $35$ terms and $13$ plus operators.

\subsection{Time and Space Expenses of the Method}

As noted above, the optimal reduction method is equivalent to applying the
DFS method to four subgraphs which are revealed by decomposing an $n$-vertex 
$FG$ in the middle. The DFS method has to be employed twice in the direct
and twice in the opposite direction. For this reason, we will implement the
optimal reduction method using the algorithms presented in section \ref%
{dfs_anal} as its base.\medskip

$FG\_Reduction\_Optimal(n,Expr)$

\begin{enumerate}
\item $\mathbf{if}$ $n>2$

\item \label{redo2}\qquad $FG\_DFS\_opposite(1,\left\lceil \frac{n}{2}%
\right\rceil ,Expr)$

\item \qquad $\mathbf{if}$ $\left\lceil \frac{n}{2}\right\rceil >2$

\item $\qquad \qquad \mathbf{Insert\_to\_Head}\left( "(",Expr\right) $

\item $\qquad \qquad \mathbf{Insert\_to\_End}(")",Expr)$

\item \label{redo6}\qquad $FG\_DFS\_direct(\left\lceil \frac{n}{2}%
\right\rceil ,n,Expr2)$

\item \qquad $\mathbf{if}$ $\left\lfloor \frac{n}{2}\right\rfloor >1$

\item \qquad \qquad $\mathbf{Insert\_to\_Head}\left( "(",Expr2\right) $

\item $\qquad \qquad \mathbf{Insert\_to\_End}(")",Expr2)$

\item \qquad $\mathbf{Concatenate}(Expr,Expr2)$

\item \qquad $\mathbf{Insert\_to\_End}\left( "+",Expr\right) $

\item \label{redo12}\qquad $FG\_DFS\_opposite(1,\left\lceil \frac{n}{2}%
\right\rceil -1,Expr2)$

\item \qquad $\mathbf{if}$ $\left\lceil \frac{n}{2}\right\rceil >3$

\item \qquad \qquad $\mathbf{Insert\_to\_Head}\left( "(",Expr2\right) $

\item $\qquad \qquad \mathbf{Insert\_to\_End}(")",Expr2)$

\item \qquad $\mathbf{Concatenate}(Expr,Expr2)$

\item \qquad $\mathbf{Insert\_to\_End}\left( "b_{\left\lceil \frac{n}{2}%
\right\rceil -1}",Expr\right) $

\item \label{redo18}\qquad $FG\_DFS\_direct(\left\lceil \frac{n}{2}%
\right\rceil +1,n,Expr2)$

\item \qquad $\mathbf{if}$ $\left\lfloor \frac{n}{2}\right\rfloor >2$

\item \qquad \qquad $\mathbf{Insert\_to\_Head}\left( "(",Expr2\right) $

\item $\qquad \qquad \mathbf{Insert\_to\_End}(")",Expr2)$

\item \qquad $\mathbf{Concatenate}(Expr,Expr2)$

\item $\mathbf{else}$

\item $\qquad \mathbf{if}$ $n$ $=$ $2$

\item $\qquad \qquad Expr\leftarrow "a_{1}"$

\item $\qquad \mathbf{else}$

\item $\qquad \qquad Expr\leftarrow \mathbf{NULL}$
\end{enumerate}

For even $n$, the number of joint reductions implemented by this algorithm
will be greater than the number of fork reductions by one. Hence, there
exists the second version of the algorithm with another choice of the middle
of the graph for even $n$ which implements the greater number of fork
reductions. Both realizations give the same expression for odd $n$.

Running times of procedures for the DFS method applied to an $\frac{n}{2}$%
-vertex $FG$ (lines \ref{redo2}, \ref{redo6}, \ref{redo12}, \ref{redo18})
are, as follows from section \ref{dfs_anal}, $\Theta \left( \left( \frac{1+%
\sqrt{5}}{2}\right) ^{n/2}\right) $. Other operations of the algorithm are
performed in $O(1)$ time. Thus, the running time of the algorithm is $\Theta
\left( \left( \sqrt{\frac{1+\sqrt{5}}{2}}\right) ^{n}\right) $.

The amount of memory that requires the algorithm is determined only by the
size of the derived expression and, by Corollary \ref{cor_red_exp}, is also $%
\Theta \left( \left( \sqrt{\frac{1+\sqrt{5}}{2}}\right) ^{n}\right) $.

\section{Decomposition Method\label{sec_fg_dec}}

This method is based on revealing subgraphs in the initial graph. The
resulting expression is produced by a special composition of subexpressions
describing these subgraphs.

Consider the $n$-vertex $FG$ presented in Figure \ref{fig2}. Denote by $%
E(p,q)$ a subexpression related to its subgraph (which is an $FG$ as well)
having a source $p$ ($1\leq p\leq n$) and a sink $q$ ($1\leq q\leq n$, $%
q\geq p$). If $q-p\geq 2$, then we choose any \textit{decomposition vertex} $%
i$ ($p+1\leq i\leq q-1$) in a subgraph, and, in effect, split it at this
vertex (Figure \ref{fg_fig7}). Otherwise, we assign final values to $E(p,q)$%
. As follows from the $FG$ structure, any path from vertex $p$ to vertex $q$
passes through vertex $i$ or avoids it via edge $b_{i-1}$. Therefore, $%
E(p,q) $ can be generated by the following recursive procedure (\textit{%
decomposition procedure}):

\begin{enumerate}
\item \label{1}$\mathbf{case}$ $q=p:E(p,q)\leftarrow 1$

\item \label{2}$\mathbf{case}$ $q=p+1:E(p,q)\leftarrow a_{p}$

\item \label{3}$\mathbf{case}$ $q\geq p+2:\mathbf{choice}(p,q,i)$

\item \label{4}$\qquad \qquad \qquad \quad ~E(p,q)\leftarrow
E(p,i)E(i,q)+E(p,i-1)b_{i-1}E(i+1,q)$
\end{enumerate}

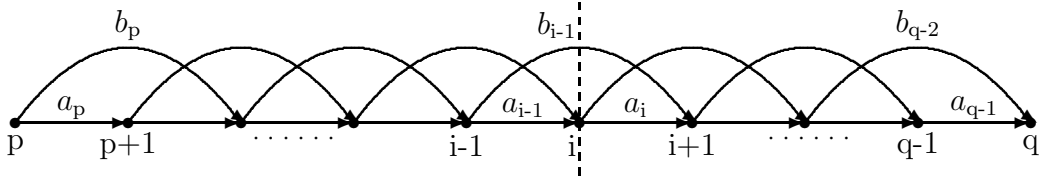
\begin{figure}[tbp]
\setlength{\unitlength}{1.0cm}
\par
\begin{picture}(5,2)(-0.5,-0.4)\thicklines

\multiput(0,0)(1.5,0){10}{\circle*{0.15}}

\put(0,-0.3){\makebox(0,0){p}}
\put(1.5,-0.3){\makebox(0,0){p+1}}
\put(6,-0.3){\makebox(0,0){i-1}}
\put(7.4,-0.3){\makebox(0,0){i}}
\put(9,-0.3){\makebox(0,0){i+1}}
\put(12,-0.3){\makebox(0,0){q-1}}
\put(13.5,-0.3){\makebox(0,0){q}}

\multiput(0,0)(1.5,0){9}{\vector(1,0){1.5}}

\put(0.75,0.2){\makebox(0,0){$a_{\text{p}}$}}
\put(6.75,0.2){\makebox(0,0){$a_{\text{i-1}}$}}
\put(8.25,0.2){\makebox(0,0){$a_{\text{i}}$}}
\put(12.75,0.2){\makebox(0,0){$a_{\text{q-1}}$}}

\qbezier(0,0)(1.5,2)(3,0)
\qbezier(1.5,0)(3,2)(4.5,0)
\qbezier(3,0)(4.5,2)(6,0)
\qbezier(4.5,0)(6,2)(7.5,0)
\qbezier(6,0)(7.5,2)(9,0)
\qbezier(7.5,0)(9,2)(10.5,0)
\qbezier(9,0)(10.5,2)(12,0)
\qbezier(10.5,0)(12,2)(13.5,0)

\put(1.5,1.3){\makebox(0,0){$b_{\text{p}}$}}
\put(7.2,1.3){\makebox(0,0){$b_{\text{i-1}}$}}
\put(12,1.3){\makebox(0,0){$b_{\text{q-2}}$}}

\multiput(3.085,0)(1.5,0){8}{\vector(3,-2){0}}

\multiput(3.2,-0.2)(0.2,0){6}{\circle*{0.02}}
\multiput(10.05,-0.2)(0.2,0){6}{\circle*{0.02}}

\multiput(7.5,-0.7)(0,0.2){12}{\line(0,1){0.1}}

\end{picture}
\caption{Decomposition of a Fibonacci subgraph at vertex $i$.}
\label{fg_fig7}
\end{figure}

Lines \ref{1} and \ref{2} contain conditions of exit from the recursion. The
special case when a subgraph consists of a single vertex is considered in
line \ref{1}. It is clear that such a subgraph can be connected to other
subgraphs only serially. For this reason, it is accepted that its
subexpression is $1$, so that when it is multiplied by another
subexpression, the final result is not influenced. Line \ref{2} describes a
subgraph consisting of a single edge. The corresponding subexpression
consists of a single term equal to the edge label. The general case is
processed in lines \ref{3} and \ref{4}. The procedure, $\mathbf{choice}%
(p,q,i)$, in line \ref{3} chooses an arbitrary decomposition vertex $i$ on
the interval $(p,q)$ so that $p<i<q$. A current subgraph is decomposed into
four new subgraphs in line \ref{4}. Subgraphs described by subexpressions $%
E(p,i)$ and $E(i,q)$ include all paths from vertex $p$ to vertex $q$ passing
through vertex $i$. Subgraphs described by subexpressions $E(p,i-1)$ and $%
E(i+1,q)$ include all paths from vertex $p$ to vertex $q$ passing through
edge $b_{i-1}$.

$E(1,n)$ is the expression of the initial $n$-vertex $FG$ ($Ex\left(
FG\right) $). Hence, the decomposition procedure is initially invoked by
substituting parameters $1$ and $n$ instead of $p$ and $q$, respectively.

In \cite{KoL} we proved the following theorem that determines an optimal
location of the decomposition vertex $i$ in an arbitrary interval $(p,q)$ of
a Fibonacci graph from the perspective of the first complexity
characteristic.

\begin{theorem}
\label{th_fg-n/2}The representation with a minimum total number of terms
among all possible representations of $Ex(FG)$ derived by the decomposition
method is achieved if and only if in each recursive step $i$ is equal to $%
\frac{q+p}{2}$ for odd $q-p+1$ and to $\frac{q+p-1}{2}$ or $\frac{q+p+1}{2}$
for even $q-p+1$, i.e., when $i$ is a middle vertex of the interval $(p,q)$.
Such a decomposition method is called \textit{optimal}.
\end{theorem}

The following theorem for the second complexity characteristic is proven in 
\cite{Kor}.

\begin{theorem}
\label{th_fg-n/2_P}The representation with a minimum number of plus
operators among all possible representations of $Ex(FG)$ derived by the
decomposition method can be achieved by the optimal decomposition method.
\end{theorem}

It can be easily shown that for an $n$-vertex $FG$:\medskip

1. The total number of terms $T(n)$ in the expression $Ex(FG)$ derived by
the optimal decomposition method is defined recursively as follows: 
\begin{eqnarray*}
T(1) &=&0 \\
T(2) &=&1 \\
T(n) &=&T\left( \left\lceil \frac{n}{2}\right\rceil \right) +T\left(
\left\lfloor \frac{n}{2}\right\rfloor +1\right) +T\left( \left\lceil \frac{n%
}{2}\right\rceil -1\right) +T\left( \left\lfloor \frac{n}{2}\right\rfloor
\right) +1\quad (n>2).
\end{eqnarray*}

2. The number of plus operators $P(n)$ in the expression $Ex(FG)$ derived by
the optimal decomposition method is defined recursively as follows: 
\begin{eqnarray*}
P(1) &=&0 \\
P(2) &=&0 \\
P(n) &=&P\left( \left\lceil \frac{n}{2}\right\rceil \right) +P\left(
\left\lfloor \frac{n}{2}\right\rfloor +1\right) +P\left( \left\lceil \frac{n%
}{2}\right\rceil -1\right) +P\left( \left\lfloor \frac{n}{2}\right\rfloor
\right) +1\quad (n>2).
\end{eqnarray*}%
For large $n$%
\begin{equation*}
T(n)\approx 4T\left( \left\lceil \frac{n}{2}\right\rceil \right) +1.
\end{equation*}

By the \textit{master theorem} \cite{CLR} recurrences like 
\begin{equation*}
T(n)=\alpha T\left( \frac{n}{\beta }\right) +f(n),
\end{equation*}%
where $\alpha \geq 1$ and $\beta >1$ are constants, and $\frac{n}{\beta }$
is interpreted as either $\left\lfloor \frac{n}{\beta }\right\rfloor $ or $%
\left\lceil \frac{n}{\beta }\right\rceil $, can be bounded asymptotically as
follows: 
\begin{equation*}
T(n)=\Theta \left( n^{\log _{\beta }\alpha }\right)
\end{equation*}%
if $f(n)=O\left( n^{\log _{\beta }\alpha -\epsilon }\right) $ for some
constant $\epsilon >0$.

Therefore, $T(n)$ and $P(n)$ are $\Theta \left( n^{2}\right) $.

For $n=9$, the possible algebraic expression derived by the optimal
decomposition method is 
\begin{eqnarray*}
&&((a_{1}a_{2}+b_{1})(a_{3}a_{4}+b_{3})+a_{1}b_{2}a_{4})((a_{5}a_{6}+b_{5})(a_{7}a_{8}+b_{7})+a_{5}b_{6}a_{8})+
\\
&&(a_{1}(a_{2}a_{3}+b_{2})+b_{1}a_{3})b_{4}(a_{6}(a_{7}a_{8}+b_{7})+b_{6}a_{8}).
\end{eqnarray*}%
It contains $31$ terms and $11$ plus operators.

We conjecture that the optimal decomposition method provides an optimal
representation (for both our complexity characteristics) of an algebraic
expression related to a Fibonacci graph.

As shown in \cite{Kor}, the optimal decomposition method is not always the
only one that provides an expression for a Fibonacci graph with a minimum
number of plus operators. There exist \textit{special values of }$n$ when an 
$n$-vertex Fibonacci graph has several expressions with the same minimum
number of plus operators (among expressions derived by the decomposition
method). These special values are grouped as follows: 
\begin{equation*}
7,13\div 15,25\div 31,49\div 63,97\div 127,193\div 255,\ldots
\end{equation*}%
In the general view, they can be presented in the following way: 
\begin{eqnarray*}
n_{first_{\nu }} &\leq &n_{sp_{\nu }}\leq n_{last_{\nu }}, \\
n_{first_{1}} &=&n_{last_{1}}=7, \\
n_{first_{\nu }} &=&2n_{first_{\nu -1}}-1, \\
n_{last_{\nu }} &=&2n_{last_{\nu -1}}+1.
\end{eqnarray*}%
Here $\nu $ is a number of a group of special numbers; $n_{sp_{\nu }}$ is a
special number of the $\nu $-th group; $n_{first_{\nu }}$ and $n_{last_{\nu
}}$ are the first value and the last value, respectively, in the $\nu $-th
group. For all these values of $n$, not only the values of $i$ which are
mentioned in Theorem \ref{th_fg-n/2}, provide a minimum number of plus
operators in $Ex(FG)$.

It can be shown that if $i=3$ or $i=n-2$ in every recursive step (the same
value in each step) then the decomposition method turns out to be the DLS
method. If $i$ is an arbitrary number in the first recursive step, and in
all subsequent steps $i=2$ or $i=n-1$ (the same value in each step) then the
decomposition method may be interpreted as the reduction method.
Specifically, if $i$ in the first step is the same as $i$ in all following
steps, then the decomposition method coincides with the DFS method.

\subsection{Time and Space Expenses of the Method}

We propose the following recursive algorithm which realizes the optimal
decomposition\ method in accordance with the above procedure and Theorem \ref%
{th_fg-n/2}:\medskip\ 

$FG\_Decomposition\_Optimal(p,q,Expr)$

\begin{enumerate}
\item $\mathbf{if}$ $q-p>1$

\item \label{deco2}\qquad $FG\_Decomposition\_Optimal(p,\left\lfloor \frac{%
q+p}{2}\right\rfloor ,Expr)$

\item \qquad $\mathbf{if}$ $\left\lfloor \frac{q+p}{2}\right\rfloor -p>1$

\item $\qquad \qquad \mathbf{Insert\_to\_Head}\left( "(",Expr\right) $

\item $\qquad \qquad \mathbf{Insert\_to\_End}(")",Expr)$

\item \label{deco6}\qquad $FG\_Decomposition\_Optimal(\left\lfloor \frac{q+p%
}{2}\right\rfloor ,q,Expr2)$

\item \qquad $\mathbf{if}$ $q-\left\lfloor \frac{q+p}{2}\right\rfloor >1$

\item \qquad \qquad $\mathbf{Insert\_to\_Head}\left( "(",Expr2\right) $

\item $\qquad \qquad \mathbf{Insert\_to\_End}(")",Expr2)$

\item \qquad $\mathbf{Concatenate}(Expr,Expr2)$

\item \qquad $\mathbf{Insert\_to\_End}\left( "+",Expr\right) $

\item \label{deco12}\qquad $FG\_Decomposition\_Optimal(p,\left\lfloor \frac{%
q+p}{2}\right\rfloor -1,Expr2)$

\item \qquad $\mathbf{if}$ $\left\lfloor \frac{q+p}{2}\right\rfloor -p>2$

\item \qquad \qquad $\mathbf{Insert\_to\_Head}\left( "(",Expr2\right) $

\item $\qquad \qquad \mathbf{Insert\_to\_End}(")",Expr2)$

\item \qquad $\mathbf{Concatenate}(Expr,Expr2)$

\item \qquad $\mathbf{Insert\_to\_End}\left( "b_{\left\lfloor \frac{q+p}{2}%
\right\rfloor -1}",Expr\right) $

\item \label{deco18}\qquad $FG\_Decomposition\_Optimal(\left\lfloor \frac{q+p%
}{2}\right\rfloor +1,q,Expr2)$

\item \qquad $\mathbf{if}$ $q-\left\lfloor \frac{q+p}{2}\right\rfloor >2$

\item \qquad \qquad $\mathbf{Insert\_to\_Head}\left( "(",Expr2\right) $

\item $\qquad \qquad \mathbf{Insert\_to\_End}(")",Expr2)$

\item \qquad $\mathbf{Concatenate}(Expr,Expr2)$

\item $\mathbf{else}$

\item $\qquad \mathbf{if}$ $q-p$ $=$ $1$

\item $\qquad \qquad Expr\leftarrow "a_{p}"$

\item $\qquad \mathbf{else}$

\item $\qquad \qquad Expr\leftarrow \mathbf{NULL}$
\end{enumerate}

The algorithm is initially invoked by substituting $1$ and $n$ instead of $p$
and $q$, respectively, for deriving the expression of an $n$-vertex $FG$.
Four recursive calls (lines \ref{deco2}, \ref{deco6}, \ref{deco12}, \ref%
{deco18}) include a middle of the interval $(p,q)$ in their parameter list ($%
\left\lfloor \frac{q+p}{2}\right\rfloor $ is a number of the decomposition
vertex while its another possible number is $\left\lceil \frac{q+p}{2}%
\right\rceil $). Hence, the running time of the algorithm is $t(n)\approx
4t\left( \left\lceil \frac{n}{2}\right\rceil \right) +O(1)=\Theta \left(
n^{2}\right) $.

The amount of memory that requires the algorithm is determined only by the
size of the derived expression and, thus, is also $\Theta \left(
n^{2}\right) $.

\section{Generalized Decomposition (GD) Method\label{sec_GD}}

As follows from the previous section, the decomposition method is based on
splitting an $FG$ in each recursive step into two parts via decomposition
vertex $i$ and edge $b_{i-1}$. The GD method entails splitting an $FG$ in
each recursive step into an arbitrary number of parts (we will denote this
number by $m$) via \textit{decomposition vertices} $i_{1},i_{2},\ldots
,i_{m-1}$ and edges $b_{i_{1}-1},b_{i_{2}-1},\ldots ,b_{i_{m-1}-1}$,
respectively. An example for $m=3$ is illustrated in Figure \ref{fg_fig8}.

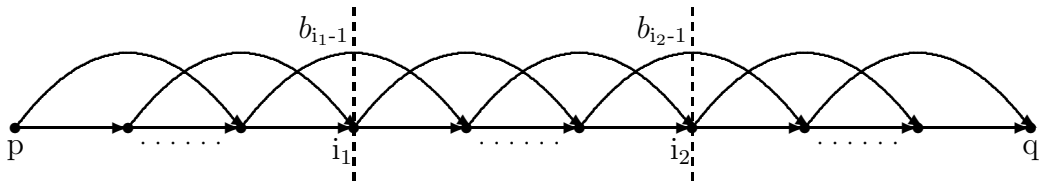
\begin{figure}[tbph]
\setlength{\unitlength}{1.0cm}
\par
\begin{picture}(5,2)(-0.5,-0.4)\thicklines

\multiput(0,0)(1.5,0){10}{\circle*{0.15}}

\put(0,-0.3){\makebox(0,0){p}}
\put(4.35,-0.3){\makebox(0,0){i$_{\text{1}}$}}
\put(8.85,-0.3){\makebox(0,0){i$_{\text{2}}$}}
\put(13.5,-0.3){\makebox(0,0){q}}

\multiput(0,0)(1.5,0){9}{\vector(1,0){1.5}}

\qbezier(0,0)(1.5,2)(3,0)
\qbezier(1.5,0)(3,2)(4.5,0)
\qbezier(3,0)(4.5,2)(6,0)
\qbezier(4.5,0)(6,2)(7.5,0)
\qbezier(6,0)(7.5,2)(9,0)
\qbezier(7.5,0)(9,2)(10.5,0)
\qbezier(9,0)(10.5,2)(12,0)
\qbezier(10.5,0)(12,2)(13.5,0)

\put(4.1,1.3){\makebox(0,0){$b_{\text{i}_{\text{1}}\text{-1}}$}}
\put(8.6,1.3){\makebox(0,0){$b_{\text{i}_{\text{2}}\text{-1}}$}}

\multiput(3.085,0)(1.5,0){8}{\vector(3,-2){0}}

\multiput(1.7,-0.2)(0.2,0){6}{\circle*{0.02}}
\multiput(6.2,-0.2)(0.2,0){6}{\circle*{0.02}}
\multiput(10.7,-0.2)(0.2,0){6}{\circle*{0.02}}

\multiput(4.5,-0.7)(0,0.2){12}{\line(0,1){0.1}}
\multiput(9,-0.7)(0,0.2){12}{\line(0,1){0.1}}

\end{picture}
\caption{Decomposition of a Fibonacci subgraph at vertices $i_{1}$ and $%
i_{2} $.}
\label{fg_fig8}
\end{figure}

In all cases when $m>2$, the decomposition procedure used in the previous
section is transformed to the more complex form. Specifically, for $m=3$,
the general line of the new decomposition procedure, corresponding to line %
\ref{4} of the decomposition procedure with $m=2$ is presented as: 
\begin{eqnarray*}
~E(p,q) &\leftarrow &E(p,i_{1})E(i_{1},i_{2})E(i_{2},q)+ \\
&&E(p,i_{1}-1)b_{i_{1}-1}E(i_{1}+1,i_{2})E(i_{2},q)+ \\
&&E(p,i_{1})E(i_{1},i_{2}-1)b_{i_{2}-1}E(i_{2}+1,q)+ \\
&&E(p,i_{1}-1)b_{i_{1}-1}E(i_{1}+1,i_{2}-1)b_{i_{2}-1}E(i_{2}+1,q).
\end{eqnarray*}%
The sum above consists of four parts, with each part including three
subexpressions corresponding to the three parts of a split subgraph. Hence,
a current subgraph is decomposed into twelve new subgraphs.

Suppose that an $FG$ is split into approximately equal parts in each
recursive step (distances between decomposition vertices are equal or
approximately equal). It will be the \textit{uniform GD method}.

The following theorem is proven in \cite{KoL1}.

\begin{theorem}
\label{th_gd}For an $n$-vertex $FG$, both the total number of terms $T(n)$
and the number of plus operators $P(n)$ in the expression $Ex(FG)$ derived
by the uniform GD method (the $FG$ is split into $m$ parts) are $O\left(
n^{1+\log _{m}2^{m-1}}\right) $.
\end{theorem}

As follows from Theorem \ref{th_gd}, $T(n)$ and $P(n)$ reach the minimum
complexity among $2\leq m\leq n-1$ when $m=2$. Substituting $2$ for $m$
gives $O\left( n^{2}\right) $ (we have the optimal decomposition method in
this case). Further, the complexity increases with the increase in $m$. For
example, we have $O\left( n^{1+\log _{3}4}\right) $ for $m=3$, $O\left(
n^{2.5}\right) $ for $m=4$, etc. In the extreme case, when $m=n-1$, all
inner vertices (from $2$ to $n-1$) of an $n$-vertex $FG$ are decomposition
vertices. The single recursive step is executed in this case, and all
revealed subgraphs are individual edges (labeled $a$ with index) connected
by additional edges (labeled $b$ with index). That is, in this instance, the
uniform GD method is reduced to the sequential paths method. Substituting $%
n-1$ for $m$ gives 
\begin{equation*}
O\left( n^{1+\log _{n-1}2^{n-2}}\right) >O\left( n^{1+\log
_{n}2^{n-2}}\right) =O\left( 2^{n-2}n\right) .
\end{equation*}

These results do not contradict our conjecture that the optimal
decomposition method provides an optimal representation. At least, it is the
best one among uniform GD methods (asymptotically).

\section{Conclusions}

Various algorithms generating an algebraic expression for a Fibonacci graph
were proposed. Complexities of representations derived by sequential paths,
DFS, DLS, and reduction methods increase exponentially as the number of the
graph's vertices increases. The generalized decomposition (GD) method has
algorithms generating representations with polynomial complexity.
Specifically, the decomposition method provides an algorithm for
constructing the expression with $O\left( n^{2}\right) $ complexity. The
methods we considered are closely related one to another. The GD method
encompasses the widest class of algorithms, and among them, all algorithms
of the decomposition method. One of them, the optimal decomposition method,
is assumed to be the best, from the perspective of complexity. The
decomposition algorithms class comprises as subclasses reduction algorithms,
two DLS algorithms, etc. The subclass of reduction algorithms includes the
optimal reduction method, two DFS algorithms, etc. The DFS method is the
worst among reduction algorithms and is assumed to be the worst among
decomposition algorithms. The sequential paths method is a special case of
the GD method, but it does not belong to the class of decomposition
algorithms. It generates an expression with the maximum complexity. The
different methods relationship is illustrated in Figure \ref{fg_fig2}. 
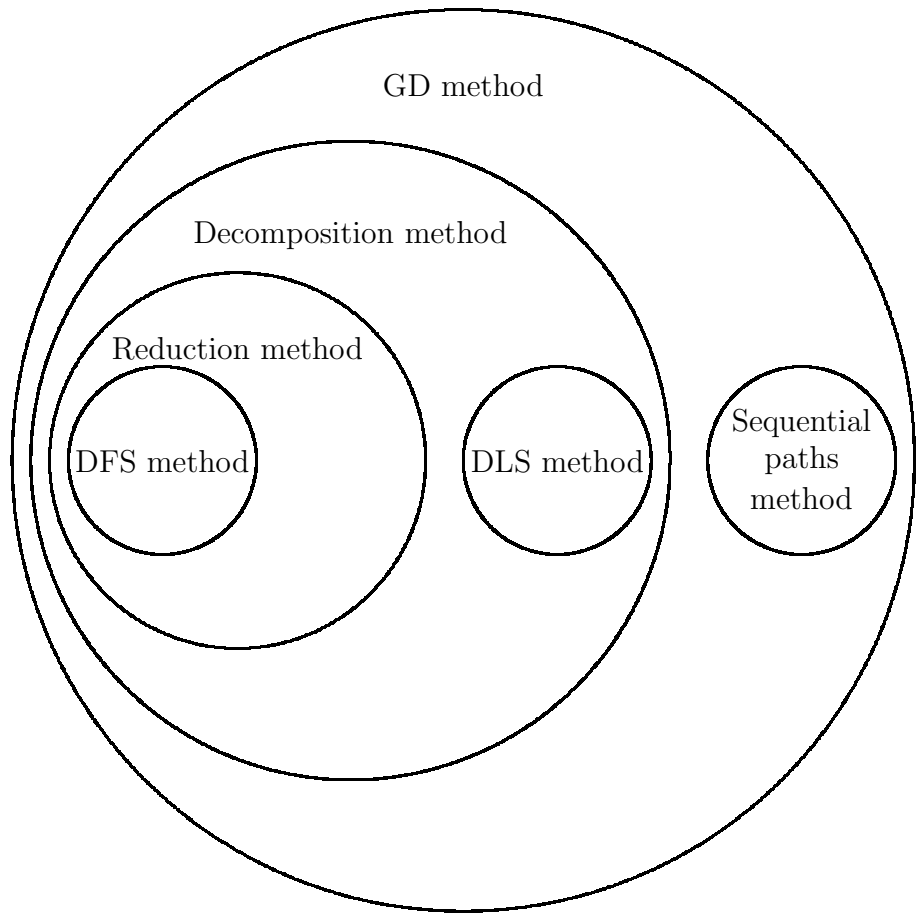
\begin{figure}[t]
\setlength{\unitlength}{1.0cm}
\par
\begin{picture}(10,10.5)(0,0)\thicklines

\put(7,6){\bigcircle[100]{12}}
\put(5.5,6){\bigcircle[100]{8.5}}
\put(11.5,6){\bigcircle[50]{2.5}}
\put(4,6){\bigcircle[50]{5}}
\put(3,6){\bigcircle[50]{2.5}}
\put(8.25,6){\bigcircle[50]{2.5}}

\put(7,11){\makebox(0,0){GD method}}
\put(5.5,9){\makebox(0,0){Decomposition method}}
\put(11.5,6.5){\makebox(0,0){Sequential}}
\put(11.5,6){\makebox(0,0){paths}}
\put(11.5,5.5){\makebox(0,0){method}}
\put(4,7.5){\makebox(0,0){Reduction method}}
\put(3,6){\makebox(0,0){DFS method}}
\put(8.25,6){\makebox(0,0){DLS method}}

\end{picture}
\caption{Relationships among methods of an algebraic expression generation
for a Fibonacci graph.}
\label{fg_fig2}
\end{figure}


\begin{thebibliography}{99}
\bibitem{BKS} W. W. Bein, J. Kamburowski, and M. F. M. Stallmann, \emph{%
Optimal Reduction of Two-Terminal Directed Acyclic Graphs}, SIAM Journal of
Computing, Vol. \textbf{21}, No \textbf{6}, 1992, 1112--1129.

\bibitem{CLR} Th. H. Cormen, Ch. E. Leiseron, and R. L. Rivest, \emph{%
Introduction to Algorithms}, The MIT Press, Cambridge, Massachusetts, 1994.

\bibitem{Duf} R. J. Duffin, \emph{Topology of Series-Parallel Networks},
Journal of Mathematical Analysis and Applications \textbf{10}, 1965,
303--318.

\bibitem{GoM} M. Ch. Golumbic and A. Mintz,\textit{\ }\emph{Factoring Logic
Functions Using Graph Partitioning}, in: Proc. IEEE/ACM Int. Conf. Computer
Aided Design, November 1999, 109--114.

\bibitem{GMR} M. Ch. Golumbic, A. Mintz, and U. Rotics,\textit{\ }\emph{%
Factoring and Recognition of Read-Once Functions using Cographs and Normality%
}, in: Proc. 38th Design Automation Conf., June 2001, 195--198.

\bibitem{GoP} M. Ch. Golumbic and Y. Perl,\textit{\ }\emph{Generalized
Fibonacci Maximum Path Graphs}, Discrete Mathematics \textbf{28}, 1979,
237--245.

\bibitem{Kor} M. Korenblit, \emph{Efficient Computations on Networks}, Ph.D.
Thesis, Bar-Ilan University, Israel, 2004.

\bibitem{KoL} M. Korenblit and V. E. Levit, \emph{On Algebraic Expressions
of Series-Parallel and Fibonacci Graphs}, in: Discrete Mathematics and
Theoretical Computer Science, Proc. 4th Int. Conf., DMTCS 2003, LNCS \textbf{%
2731}, Springer, 2003, 215--224.

\bibitem{KoL1} M. Korenblit and V. E. Levit, \emph{The Uniform Generalized
Decomposition Method for Generating Algebraic Expressions of Fibonacci Graphs%
}, WSEAS Transactions on Mathematics, Vol. \textbf{2}, No 1, 2003, 92--97

\bibitem{KoL2} M. Korenblit and V. E. Levit, \emph{On Algebraic Expressions
of Generalized Fibonacci Graphs}, WSEAS Transactions on Mathematics, Vol. 
\textbf{2}, No 4, 2003, 324--329.

\bibitem{Mun1} D. Mundici, \emph{Functions Computed by Monotone Boolean
Formulas with no Repeated Variables}, Theoretical Computer Science \textbf{66%
}, 1989, 113-114.

\bibitem{Mun2} D. Mundici, \emph{Solution of Rota's Problem on the Order of
Series-Parallel Networks}, Advances in Applied Mathematics \textbf{12},
1991, 455--463.

\bibitem{Nau} V. Naumann, \emph{Measuring the Distance to Series-Parallelity
by Path Expressions}, in: Graph-Theoretic Concepts in Computer Science,
Proc. 20th Int. Workshop, WG '94, LNCS \textbf{903}, Springer, 1994,
269--281.

\bibitem{Ros} \emph{Handbook of Discrete and Combinatorial Mathematics},
edited by K. H. Rosen, CRC Press, Boca Raton, 2000.

\bibitem{SaW} P. Savicky and A. R. Woods, \emph{The Number of Boolean
Functions Computed by Formulas of a Given Size}, Random Structures and
Algorithms \textbf{13}, 1998, 349--382.

\bibitem{Wan} A. R. R. Wang, \emph{Algorithms for Multilevel Logic
Optimization}, Ph.D. Thesis, University of California, Berkeley, 1989.
\end{thebibliography}
\end{document}